\documentclass{lmcs} 

\keywords{unification, convergent string-rewriting systems, fixed point problem, common term problem, common equation problem, conjugacy problem, common multiplier problem}

\usepackage{hyperref}
\theoremstyle{plain} 

\usepackage{algorithm}
\usepackage{algpseudocode}
\usepackage{amsthm}
\usepackage{amssymb} 
\usepackage{amsmath} 
\usepackage{wasysym} 
\usepackage{epsfig, graphics, graphicx} 
\usepackage[toc,page]{appendix}
\usepackage{mathtools} 
\usepackage{verbatim}

\usepackage{fancybox}

\usepackage{multirow}

\usepackage{tikz}
\usetikzlibrary{automata,positioning,arrows}
\usetikzlibrary{shapes}
\newcommand\irregularcircle[2]{
  \pgfextra {\pgfmathsetmacro\len{(#1)+rand*(#2)}}
  +(0:\len pt)
  \foreach \a in {10,20,...,350}{
    \pgfextra {\pgfmathsetmacro\len{(#1)+rand*(#2)}}
    -- +(\a:\len pt)
  } -- cycle
}

\tikzset{set/.style={draw,circle,inner sep=0pt,align=center}}
\usepackage{caption}
\usepackage{subcaption}

\newcommand{\ignore}[1]{}

\newcommand{\flr}{\rightarrow}

\newcommand{\Cross}{\mathbin{\tikz [x=1.4ex,y=1.4ex,line width=.2ex] \draw (0,0) -- (1,1) (0,1) -- (1,0);}}%

\newcommand{\Var}{\mathcal{V}\!\mathit{ar}}
\newcommand{\dom}{\mathcal{D}om}

\newcommand{\vran}{\mathcal{V\!R}\!\mathit{an}}

\newtheorem{defn}{Definition}

\newcommand{\normform}[1]{(#1)_\downarrow}

\newcommand{\strind}[1]{_{\mathsf{[\, #1 \,]}}}

\newcommand{\ssp}[1]{^{(#1)}}

\newcommand{\raw}{\rightarrow}

\newcommand{\rawrnorm}[1][R]{\rightarrow_{#1}^!}

\newcommand{\rawrrtc}[1][R]{\rightarrow_{#1}^*}

\newcommand{\drawr}[1][R]{\rightsquigarrow_{#1}}
\newcommand{\drawrnorm}[1][R]{\rightsquigarrow_{#1}^!}
\newcommand{\drawrtc}[1][R]{\rightsquigarrow_{#1}^+}
\newcommand{\drawrrtc}[1][R]{\rightsquigarrow_{#1}^*}

\def\presuper#1#2%
  {\mathop{}%
   \mathopen{\vphantom{#2}}^{#1}%
   \kern-\scriptspace%
   #2}

\def\presub#1#2%
  {\mathop{}%
   \mathopen{\vphantom{#2}}_{#1}%
   \kern-\scriptspace%
   #2}

\def\presubsuper#1#2#3%
  {\mathop{}%
   \mathopen{\vphantom{#3}}_{#1}^{#2}%
   \kern-\scriptspace%
   #3}

\makeatletter
\newcommand*{\rom}[1]{\expandafter\@slowromancap\romannumeral #1@}
\makeatother

\begin{document}

\title[On Problems Dual to Unification: The String-Rewriting Case]{On Problems Dual to Unification:\\
\large The String-Rewriting Case}
\titlecomment{{\lsuper*}A variant of the paper has also been published in \cite{zumrutDissertation, 2017arXiv, UNIF2017Akcam}. Some of the results reported here are a partial fulfillment of the Ph.D. requirements of the fourth author, and will be part of his dissertation.}

\author[Akcam]{Z\"{u}mr\"{u}t Ak{\c c}am}	
\address{Stevens Institute of Technology, Hoboken, NJ, US}	
\email{zakcamki@stevens.edu}  
\thanks{Thanks to Dr. Daniel J. Dougherty for his feedback}	

\author[Cornell]{Kimberly A. Cornell}
\address{University at Albany, SUNY\\ Albany, NY, US}	
\email{kacornell@albany.edu}  
\author[Hono]{Daniel S. Hono II}	
\address{University at Albany, SUNY\\ Albany, NY, US}	
\email{dhono@albany.edu}  

\author[Narendran]{Paliath Narendran}	
\address{University at Albany, SUNY\\ Albany, NY, US}	
\email{pnarendran@albany.edu}

\author[Pulver]{Andrew Pulver}	
\address{University at Albany, SUNY\\ Albany, NY, US}	
\email{apulver@albany.edu}





\begin{abstract}
  \noindent In this paper, we investigate problems which are dual to the
  unification problem, namely the Fixed Point~(FP) problem, Common Term~(CT) problem and
  the Common Equation~(CE) problem for string rewriting
  systems (SRS). Our main motivation is computing fixed points in
  systems, such as loop invariants in programming languages. We
  show that the fixed point~(FP) problem is reducible to the common term
  problem. Our new results are:
  (i) the fixed point problem is undecidable for finite convergent
  string rewriting systems(SRS) whereas it is decidable in polynomial time for finite, convergent and dwindling string rewriting systems,
  (ii) the common term problem is undecidable
  for the class of dwindling string rewriting systems, and
  (iii) for the
  class of finite, monadic and convergent systems, the common equation
  problem is decidable in polynomial time but for the class of dwindling string rewriting systems, common equation problem is undecidable.
\end{abstract}

\maketitle
\section*{Introduction}\label{S:one}

  Unification, with or without background theories such as
associativity and commutativity, is an area of great theoretical
and practical interest. The former problem, called \emph{equational}
or \emph{semantic} unification, has been studied from several different angles.
Here we investigate some problems that can clearly be viewed as \emph{dual}
to the unification problem.
Our main motivation for this work is theoretical, but, as explained below,
we begin with a practical application that is shared by
many fields.

In every major research field, there are variables or other parameters that change over
time. These variables are modified --- increased
or decreased --- as a
result of a change in the environment. 
Computing \emph{invariants,} or expressions whose values do not change
under a transformation, is very important in many areas 
such as Physics, e.g., invariance under the \emph{Lorentz} transformation.

In Computer Science, the issue of obtaining invariants arises
in \emph{axiomatic semantics} or \emph{Floyd-Hoare semantics},
in the context of formally proving a loop to be correct. A
\emph{loop invariant} is a
condition, over the program variables, 
that holds before and after each iteration. Our research
is partly motivated by the related question of
finding expressions, called \emph{fixed points}, whose values
will be the same before and after each iteration, i.e., will
remain unchanged as long as the iteration goes on. For instance,
for a loop whose body is \[ \text{\tt X = X + 2; Y = Y - 1;} \] the
value of the expression {\tt X + 2Y} is a fixed point.

In this paper, we explore the Fixed Point, Common Term and Common Equation problems in \emph{convergent} string rewriting systems (SRS), more specifically in the subclass of \emph{dwindling} string rewriting systems. 

We prove that the Fixed Point problem is undecidable for convergent SRS, yet polynomial for dwindling SRS. In the case of the Common Term and Common Equation problems, we will demonstrate their undecidability within the context of dwindling SRS. Additionally, we will showcase that the Common Equation problem attains polynomial complexity when considering monadic SRS. Our results and the current literature for these three problems are  summarized in Table~\ref{complexityTable}, the \emph{rectangle} boxes indicating \emph{our own
results}:

\begin{table}[h]
\begin{center}
\begin{tabular}{|c|c|c|c|c|} \hline
	& Convergent & Length-reducing & Dwindling & Monadic \\[6pt] \hline
    &            &                 &           &         \\
FP  & \doublebox{\textbf{undecidable}} & NP-complete   & \doublebox{\textbf{P}}       &  P      \\
    &            &                 &           &         \\
CT  & undecidable & undecidable    & \doublebox{\textbf{undecidable}} & P \\
    &            &                 &           &         \\
CE  & undecidable & undecidable    & \doublebox{\textbf{undecidable}} & \doublebox{\textbf{P}} \\
    &            &                 &           &         \\ \hline
\end{tabular}
\end{center}
\caption{Complexity results of the problems in String Rewriting Systems.}
\label{complexityTable}
\end{table}

An important contribution of this paper is completing the complexity results for these three problems for convergent, length-reducing, dwindling and monadic SRS. 
Length-reducing and monadic string rewriting systems are thoroughly investigated by \cite{Otto1986}, \cite{NOW}, \cite{OND98}, \cite{NarendranOtto97}. The explanation of these different subclasses of string rewriting systems can be found under the section ``Notation and Preliminaries."~\ref{S:defns}

Dwindling convergent systems are especially important because they
are a special case of \emph{subterm-convergent theories} which are
widely used in the field of protocol analysis~\cite{abadi2006deciding,
  Baudet2005, ciocaba2009, cortier2009}. Tools such as TAMARIN
subterm-convergent theories since these theories have nice properties
(e.g., finite basis property~\cite{chevalier2010compiling}) and
decidability results~\cite{abadi2006deciding}.

\section*{Notation and Preliminaries}\label{S:defns}
We start by presenting some notation and definitions on term rewriting
systems and particularly string rewriting systems. Only some
definitions are given in here, but for more details, refer to the
books~\cite{Term} for term rewriting systems and~\cite{Botto} for
string rewriting systems.

A signature $\Sigma$ consists of finitely many ranked function symbols.
Let $X$~be a (possibly infinite) set of variables. The set
of all terms over~$\Sigma$ and~$X$ is denoted~as $T(\Sigma, X)$. 
$\Var(t)$ is the set of variables for term $t$ and a term is a ground term iff $\Var(t) = \emptyset$. 
The set of \emph{ground terms}, or terms with no variables
is denoted~$T(\Sigma)$.
A term
rewriting system~(TRS) is a set of rewrite rules that are defined on
the signature~$\Sigma$, in the form of $l \rightarrow r$, 
where $l$ and $r$
are called the left-hand-side and right-hand-side (\emph{lhs} and \emph{rhs}) of the rule, respectively. The rewrite relation induced by a term rewriting system~$R$ is denoted by~$\rightarrow_R^{}$. 
The reflexive
and transitive closure of~$\rightarrow_R^{}$ is denoted
$\rightarrow_R^{*}$. A TRS~$R$
is called \emph{terminating} iff there is no infinite chain of terms.
A TRS $R$~is \emph{confluent} iff, for all terms $t$,
$s_1$, $s_2$, if $s_1$ and $s_2$ can be derived from~$t$, i.e.,
$s_1 \leftarrow_R^{*} t \rightarrow_R^{*} s_2$, then there exists a 
term~$t'$ such that $s_1 \rightarrow_R^{*} t' \leftarrow_R^{*} s_2$.  A TRS
$R$ is \emph{convergent\/} iff it is both terminating and confluent.

A term is \emph{irreducible} iff no rule of TRS $R$ can be
applied to that term. The set of terms that are irreducible modulo~$R$
is denoted by $IRR(R)$ and also called as terms in their \emph{normal
  form}s.  A term~$t'$ is said to be an \emph{R-normal form} of a term
$t$, iff it is irreducible and reachable from~$t$ in a finite number
of steps; this is written as $t \rightarrow_{R}^{!} t'$. Let $t_\downarrow$
denote the normal form of $t$ when $R$ is understood from the context.

String rewriting systems (SRS) are a restricted class of term rewriting
systems where all functions are unary. These unary operators, that are
defined by the symbols of a string, applied in the order in which
these symbols appear in the string, i.e., if $g, h \in \Sigma$, the
string $gh$ will be seen as the term $h(g(x))$. The set of all strings
over the alphabet $\Sigma$ is denoted by $\Sigma^*$ and the empty
string is denoted by the symbol $\lambda$. Thus the term rewriting system
$\{ p(s(x)) \rightarrow x , \; s(p(x)) \rightarrow x \}$ is equivalent to
the string-rewriting system \[ \{ sp \rightarrow \lambda , \;
ps \rightarrow \lambda \} \]

If $R$ is a string
rewriting system over alphabet $\Sigma$, then the single-step
reduction on $\Sigma^*$ can be written as:

For any $u,v \in \Sigma^*$, $u \rightarrow_R v$ iff there exists 
a rule~$l \rightarrow r
\in R$ such that 
$u = xly$ and $v = xry$ for some~$x, y \in \Sigma^*$; 
i.e., \[ {\rightarrow}_R^{} \; = \; \{ (xly, \, xry) \; \mid \; (l \rightarrow r) \in R, 
  \, x,y \in {{\Sigma}^*} \} \] 

For any string rewrite system~$R$ over $\Sigma$, the set of all
irreducible strings, $IRR(R)$, is a regular
language: in fact, $IRR(R) = \Sigma_{}^* \smallsetminus \{\Sigma_{}^*
l_1 \Sigma_{}^* \,\cup ... \cup \, \Sigma_{}^* l_n \Sigma_{}^*\}$,
where $l_1,\dots, l_n$ are the left-hand sides of the rules in~$R$.

Throughout the rest of the paper, 
$a, b, c, \dots, h$  will denote elements of the alphabet~$\Sigma$, 
and $l, r, u, v, w, x, y,z$ will denote strings over~$\Sigma$. Concepts such as {\em normal form,\/}  {\em terminating,\/} 
{\em confluent,\/} and {\em convergent\/} have the same definitions for string rewriting
systems as they have for term rewriting systems.
An SRS~$T$ is called {\em canonical\/} if and only if it is convergent and
{\em inter-reduced,\/} i.e., no lhs is a substring of another lhs and each rhs is an irreducible string.

For a string $x\in \Sigma^*$, the element at position $i$ is denoted
$x_{[i]}$, and the substring from position $i$ to position $j$
(inclusive) is denoted as $x_{[i:j]}$ where $i\leq j$ and this shall
denote the empty string when $i>j$. We will write $x_{[i:]}$ to denote
$x_{[i:|x|]}$ when it is cumbersome to use~$|x|$. Additionally, for an index sequence
$\beta = (\beta_1,\beta_2,\dots , \beta_c)$, we use $x_{[\beta]} := x_{[\beta_1]}x_{[\beta_2]}\dots x_{[\beta_c]}$.
Parenthesized superscripts shall be a general way to index elements in a sequence of
strings e.g.,  $(x^{(1)},x^{(2)},x^{(3)},\dots)$.

A  string rewrite system $T$ is said to be:
\begin{itemize}
\item[-] {\em monadic\/} iff the rhs of each rule in $T$ is either a single
 symbol or the empty string, e.g., $abc \rightarrow b$. \par 
\item[-] {\em dwindling\/} iff, for every rule $l \flr r$ in $T$, the rhs $r$ 
  is a {\em proper prefix\/} of its lhs $l$, e.g., $abc \rightarrow ab$. \par
\item[-] \emph{length-reducing} iff $| l | > | r |$ for all rules~$l \flr r$ in~$T$, e.g., $abc \rightarrow ba$.
\end{itemize}

\begin{figure}[h]
\centering
\pagestyle{empty}
\def\firstcircle{(-1,0) ellipse (1.5cm and 1cm)}
\def\secondcircle{(60:0cm) ellipse (4cm and 2cm)}
\def\thirdcircle{(0:1cm) ellipse (1.5cm and 1cm)}
\begin{tikzpicture}
    \begin{scope}[shift={(3cm,-5cm)}]
        \node (M) at (0,0){}; 
        \draw \firstcircle node[left=0.5cm of M] {$Monadic$};
        \draw \secondcircle node [above=1.0cm of M] {$\text{\emph{Length-Reducing}}$};
        \draw \thirdcircle node [right=0.3cm of M] {$Dwindling$};
    \end{scope}
\end{tikzpicture}
\caption{Some of the classes of String Rewriting Systems \protect\footnotemark}
\label{SRSets}
\end{figure}
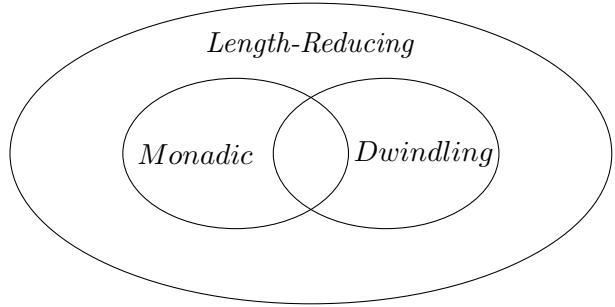

\footnotetext{The trivial forms of monadic rules such as $a \rightarrow b$ can be ignored. We can get rid of such rules by changing every occurrence of $a$ to $b$.}

\section*{Motivation and Problem Statements}
The Fixed Point, Common Term and Common Equation problems can be formulated in terms of properties of substitutions
\emph{modulo} a term rewriting system. We plan to assist the reader by providing these formulations and examples to clarify the problems for readers with no particular expertise in String Rewriting Systems.

\subsection*{Fixed Point Problem (FP)}

\begin{description}[align=left]
\item[{\bf Input}] A substitution $\theta$ and an equational theory~$E$.
\item[{\bf Question}] Does there exist a non-ground term~$t \in T(Sig(E), \, {\dom}(\theta))$ 
such that~$\theta (t) \, \approx_E^{} t$?
\end{description}

\vspace{0.05in}
\noindent
\underline{Example 1}: Suppose~$E$ is a theory of integers which contains linear
arithmetic. As an example, a similar equational theory can be found in the Example 3 below. \\ Let $\theta =\{x \mapsto x-2, \, y \mapsto y+1\}$ and we would like to find a term~$t$ such that 
$\theta (t) \approx_E^{} t$. Note that $x + 2y$ is such a term, since \[ \theta (x+2y) =
(x - 2) + 2*(y + 1) \approx_E^{} x+2y \] \\

\noindent
\underline{Example 2}: What is the fixed point/invariant of the given loop?
\begin{algorithm}
\caption{Fixed Point Loop}
	\begin{algorithmic}[1]
	    \State{} $\{x = X_0, y = Y_0\}$
	    \While{} $x > 0$
	    	\State{} $x = x -1$ \Comment \hspace{3pt} $\theta= \{x \mapsto x-1,\; y \mapsto y+x-1\}$
		\State{} $y = y + x $
	    \EndWhile
\end{algorithmic}
\end{algorithm}

Note that the value of the expression $y + \dfrac{x*(x-1)}{2}$ is unchanged, since

\begin{itemize}
\item Before Iteration: $y + \dfrac{x*(x-1)}{2}$
\item After Iteration: $$y + x - 1 + \dfrac{(x-1)*(x-2)}{2} = y + \dfrac{x^2 - 3x + 2 + 2(x-1)}{2} =  y + \dfrac{x*(x-1)}{2}$$
\end{itemize}

Thus $y + \dfrac{x*(x-1)}{2}$ is a fixed point of~$\theta$ as defined in Algorithm 1.

Note that fixed points may not be unique. Consider the term rewriting system \[ \{a(b(y)) \rightarrow a(y),\; c(b(z)) \rightarrow c(z) \} \] and
let $\theta = \{x \mapsto b(x)\}$. We can see that both $a(x)$ and $c(x)$ are fixed points of $\theta$.  

We plan to explore two related formulations, both of which can be 
viewed as \emph{dual} to the well-known unification problem. Unification
deals with solving symbolic equations: thus a typical input would be either two
terms, say $s$ and $t$, or an equation~$s \approx_{}^? t$. The task is to
find a substitution such that $\theta (s) \approx \theta(t)$.
For example, given two terms $s_1^{} = f(a, y)$ and $s_2^{} = f(x, b)$, where $f$~is
a binary function symbol, $a$ and $b$ are constants, and $x$ and $y$
are variables, the substitution $\sigma = \{ x \mapsto a, ~ y \mapsto
b\}$ unifies $s_1^{}$ and~$s_2^{}$, or equivalently, $\sigma$ is a unifier
for the equation~$s_1^{} =_{}^? s_2^{}$.

There are two ways to ``dualize'' the unification problem:\\

\vspace{1cm}
\noindent
{\large\bf Common Term Problem (CT)}:

\begin{description}[align=left]
\item [Input] Two ground substitutions $\theta_1^{}$ and $\theta_2^{}$, and an
  equational theory~$E$. (i.e., $\vran(\theta_1) = \emptyset$ and $\vran(\theta_2) = \emptyset$ )
\item [Question] Does there exist a non-ground term~$t  \in 
T(Sig(E), \, \dom(\theta_1) \cup \dom(\theta_2))$ such that
\mbox{$\theta_1^{} (t) \, \approx_E^{} \theta_2^{} (t)$}?
\end{description} 

\ignore{
We will be considering equational theories that 
are decomposable into a set of identities~$Ax$ and
a set of rewrite rules such that \emph{equational} rewriting 
modulo~$Ax$ is convergent. The most widely used case of
equational rewriting is where $Ax$ consists of associativity and
commutativity axioms~(\emph{AC}).
The key concepts are defined in the next section.
}

\vspace{0.05in}
\noindent
\underline{Example 3}: Consider the two substitutions 
\[
\theta_{1} = \{ x
\mapsto p(a), \, y \mapsto p(b) \} \textrm{ and } \theta_{2} = \{ x \mapsto a,
\, y \mapsto b \}.
\]

The following term rewriting system $R_1^{{lin}}$ specifies a fragment of linear arithmetic using
\emph{successor} and \emph{predecessor} operators:
\begin{eqnarray*}
 x - 0 & \rightarrow & x \\
 x - x & \rightarrow & 0 \\
 s(x) - y & \rightarrow & s(x-y)\\
 p(x) - y & \rightarrow & p(x-y)\\
 x - p(y) & \rightarrow & s(x-y)\\
 x - s(y) & \rightarrow & p(x-y)\\
 p(s(x)) & \rightarrow & x\\
 s(p(x)) & \rightarrow & x
\end{eqnarray*}

If we take the term rewriting system~$R_1^{lin}$ as our background equational theory~$E$, then the
common term $t =
x-y$ satisfies $\theta_{1}(t) \approx_E^{}
\theta_{2}(t)$. \[ \theta_1 ( x-y ) \approx_E^{} p(a)- p(b) \approx_E^{}
a-b \] and \[\theta_2 ( x-y ) \approx_E^{} a-b \]
\vspace{0.05in}

\noindent
We can easily show that the fixed point problem can be reduced to the CT problem, as seen below in Lemma~\ref{FPtoCTLemma}.\\[-5pt]

\begin{lem}\label{FPtoCTLemma}
The fixed point problem is reducible to the common term problem.
\end{lem}
\begin{proof}
Let $\theta_2$ be
the identity substitution. 
Assume that the fixed point problem has a solution, i.e., there exists
a term $t$ such that $\theta (t) \, \approx_E^{} t$. Then the CT problem
for $\theta$ and $\theta_2$ has a solution since
$\theta_2(t) \approx_E^{} t$ (because $\theta_2 (s) = s$ for all~$s$). The ``only if'' part is trivial, again
because $\theta_2 (s) = s$ for all~$s$.

Alternatively, suppose
that $\dom(\theta)$ consists of $n$~variables, where $n \geq 1$. If we
map all the variables in $\vran(\theta)$ to new constants, this will
create a ground substitution \[\theta_1 = \{x_1
\mapsto a_1,\; x_2 \mapsto a_2,\; ..., \;x_n \mapsto
a_n\}.\] $\theta_1$ will be the one of the substitutions for the CT
problem. The other substitution, $\theta_2$, is the composition of the
substitutions $\theta$ and $\theta_1$. The substitution $\theta_1$
will replace all of the variables in~$\vran(\theta)$ with the new
constants, thus making $\theta_2$ a ground substitution. Now
if $\theta (t) \, \approx_E^{} \, t$, then
$\theta_2 (t) = \theta_1 ( \theta (t) ) \approx_E^{} \theta_1( t )$; in other words,
$t$ is a solution to the common term problem.

The ``only if'' part can also be explained in terms of the composition
above. Suppose that $\theta_1(s)$ and $\theta_2(s)$ are equivalent,
i.e., $\theta_1(s) \, \approx_E^{} \, \theta_2(s)$ for some~$s$. Since
$\theta_2 = \theta_1 \circ \theta$, the equation can be rewritten as
$\theta_1(\theta(s)) \approx_E^{} \theta_1(s)$.
Since $a_1, \, \ldots , \, a_n$ are new constants and are not included in the signature of
the theory, for all $t_1$ and $t_2$, $\theta_1(t_1) \, \approx_E^{}
\theta_1(t_2)$ holds if and only if $t_1\approx_E^{} t_2$
(See~\cite{Term}, Section~4.1, page~60)
Thus $\theta_1(\theta(s)) \approx_E^{} \theta_1(s)$ implies that $\theta(s) \approx_E^{} s$, making
$s$~a fixed point. 
\end{proof}

\noindent
{\large\bf Common Equation Problem (CE)}:

\begin{description}[align=left]
\item[Input] Two substitutions $\theta_1^{}$ and $\theta_2^{}$ with the 
 \emph{same domain}, and an
  equational theory~$E$.
\item[Question] Does there exist a non-ground, non-trivial $(t_1
  \not\approx_E^{} t_2)$ equation~$t_1^{} \approx_{}^? t_2^{}$, where
  $t_1^{}, t_2^{} \, \in \, T(Sig(E), \, \dom(\theta_1^{}))$ such that
  both~$\theta_1^{}$ and~$\theta_2^{}$ are \emph{E}-unifiers
  of~$t_1^{} \approx_{}^? t_2^{}$?\\
\end{description} 
By trivial equations, we mean equations which are identities
in the equational theory~$E$, i.e., an equation $s \approx_{E}^? t$ is
trivial if and only if $s \approx_{E}^{} t$.
We exclude this type of trivial equations in the formulation of
this question.



\vspace{0.05in}
\noindent
\underline{Example 4}: Let $\, E ~ = ~ \{ p(s(x)) \approx x, \; s(p(x)) \approx x \}$.
Given two substitutions \[\theta_{1} = \{ x_1
\mapsto s(s(a)), \, x_2 \mapsto s(a) \}\textrm{ and }\theta_{2} = \{ x_1
\mapsto s(a), \, x_2 \mapsto a \},\] we can see that $\theta_{1}(t_1)
\approx_E^{} \theta_{1}(t_2)$ and $\theta_{2}(t_1) \approx_E^{}
\theta_{2}(t_2)$, with the equation $p(x_1) \approx_E^{} x_2$. However,
there is no term~$t$ on which the substitutions agree, i.e., there
aren't any solutions for the common term problem in this example.
Thus, CT and CE problems are not equivalent as we observe in the example above.

In the subsequent sections, this paper delves into an in-depth examination of these three problems within the realm of string rewriting systems: the Fixed Point problem, the Common Term problem, and the Common Equation problem.

\section{Fixed Point Problem}\label{S:three}
For a string rewriting system $R$, the fixed point problem can be stated as follows. 

\begin{description}[align=left]
\item[Input] A string-rewriting system~$R$ on an alphabet~$\Sigma$, and 
a string~$\alpha \in \Sigma_{}^+$.
\item[Question] Does there exist a string~$W$ such that 
$\alpha W \; \stackrel{*}{{\longleftrightarrow}_R} \; W$?
\end{description}

Since the fixed point problem is a particular case of the \emph{common term problem}, it is decidable in polynomial time for finite, monadic and convergent string rewriting systems. The fixed point problem is also a subcase of the \emph{conjugacy problem}. The conjugacy problem seeks to determine whether two given words, w and $w'$, in the group G are conjugate. In other words, the question is whether there exists a word $z$ in G such that the conjugation of 
$w$ by $z$, i.e., $zwz'$, is equal to $w'$. The conjugacy problem is both decidable in \textbf{NP} and \textbf{NP}-hard for finite, length-reducing and convergent systems. ~\cite{NOW}  

\subsection{Fixed Point Problem for Finite and Convergent Systems:}

Theorem~\ref{FPConvergent} shows that the fixed point problem is undecidable for finite and convergent string rewriting systems.

\begin{thm}
\label{lba_emptiness}
The following problem is undecidable: \begin{quote}
\begin{itemize}
\item[Input:] A non-looping, deterministic linear bounded automaton~$\mathcal{M}$ that
restores its input for \emph{accepted} strings.
\item[Question:] Is $\mathcal{L(M)}$ empty?
\end{itemize}
\end{quote}
\end{thm}

\begin{proof}

Suitably modifying the construction given in~\cite{caron1991linear} in a straightforward way, we can prove the undecidability of this problem with a reduction from the well-known
Post Correspondence Problem~(PCP). Recall that an instance of PCP is a
collection of pairs of strings over an alphabet~$\Sigma$ (``dominos'' in~\cite{Sipser}) and
the question is if there exists a sequence of
indices~$i_1, i_2, \ldots i_n$ such that
$x_{i_1}  x_{i_2} \ldots x_{i_n} = y_{i_1} y_{i_2} \ldots y_{i_n}$.
Alternatively, we can define the PCP in terms
of two homomorphisms
$\psi$ and $\phi$ from $C^*=\{c_1,\;\ldots\;c_n\}^*$ to $\Sigma^* =
\{a,b\}^*$. The question now is whether there exists a non-empty
string~$m_C^{}$ such that
$\psi(m_C^{}) = \phi(m_C^{})$.

The reduction proof is done by validating the solution string ~$w
\in \Sigma^*$, for the two homomorphisms using the Deterministic Linear Bounded Automata (DLBA) $\mathcal{M}$~\cite{kuroda1964classes}
with the following configuration, $\mathcal{M} = (Q, \Sigma, \Gamma,
\delta, q_0, q_a, q_r)$. Assume $w=w_1 \; w_2$ where $w_1 \in C^*$ and
$w_2 \in \Sigma^*$. $\mathcal{M}$ will check the correctness of $w$ by
going over one symbol at a time in the string $w_1$ and its
corresponding mapping in~$w_2$. $\mathcal{M}$ uses a marking
technique, marking the symbols~$c_i$ and their correct mappings with
overlined symbols. If every letter in the string is overlined, then
clearly the string is ``okay'' according to the first homomorphism;
$\mathcal{M}$ will then replace the overlined letters with the same
non-overlined letters as before and repeat the same steps for the
second homomorphism. Thus the DLBA will accept and re-create~$w_1 \;
w_2$ if and only if it is a solution string. $\mathcal{L(M)}$ is empty if and only if the instance of PCP has no solution.
\end{proof}

We construct a string
rewriting system~$R$ from the above-mentioned DLBA~$\mathcal{M}$.  The construction of
$R$ is similar to the one in~\cite{bauer1984finite}. Let 
$\mathcal{M} = (Q, \Sigma, \Gamma, \delta, q_0, q_a, q_r)$. Here
$\Sigma$ denotes the input alphabet, $\Gamma$ is the tape alphabet, $Q$ is the set of states, $q_0 \in
Q$ is the initial state, $q_a \in Q$ the accepting state and
$q_r \in Q$ the rejecting state. We assume
that the tape has two end-markers: $\cent \in \Gamma$ denotes
the left end-marker and $\$ \in \Gamma$ is the right 
end-marker. We also assume that on acceptance
the DLBA comes to a halt at the left-end of the tape. Finally, 
$\delta:  Q \times \Gamma \rightarrow Q \times \Gamma \times \{L, R\}$
is the transition function of~$\mathcal{M}$. 

The alphabet of $R$ is $\Gamma = \Sigma \cup \Sigma^{\prime} \cup Q$
where $\Sigma^{\prime}$ is a replica of the alphabet
$\Sigma$ such that $\Sigma^{\prime} \cap \Sigma = \emptyset$.
$R$ has the rules
\begin{equation*}
\begin{aligned}
    q_i \, a_k \, &\rightarrow  a_k^{\prime}\, q_j 			&& \text{if }(q_i, a_k, q_j, a_k^{\prime}, R) \in \delta\\
    a_l^{\prime} \, q_i \, a_k \, &\rightarrow q_j \, a_l \, a_k	&& \text{for all } a_l^{\prime} \in \Sigma^{\prime} \text{, if }(q_i, a_k, q_j, a_l,  L) \in \delta\\
    q_a \, \cent \, &\rightarrow \cent 					&& \text{ }
\end{aligned}
\end{equation*}

Since the linear bounded automaton is deterministic, $R$ is locally
confluent. Besides, $\mathcal{M}$ ultimately always halts, and that means
there will be no infinite chain of rewrites for~$R$, and thus $R$ is
terminating.

\begin{lem}\label{initialAccepting}
$\mathcal{M}$ accepts $w$ iff
$q_0\, \cent \, w \, \$ \rightarrow_R^{+} q_a\, \cent \, w\, \$ $.
\end{lem}
\begin{proof}
By inspection of the rules we can see that $\mathcal{M}$ makes the transition
$u_1\,q_1\,v_1 ~ \vdash_M^{} ~ {u_2}\, q_2\, v_2$ if and only if
$u_1^{\prime}\,q_1\,v_1 \rightarrow_R^{} u_2^{\prime}\, q_2\, v_2$.
\end{proof}

\begin{lem}\label{initialFixedPoint}
$\cent \,w\, \$$ is a fixed point for $q_0$ in $R$ iff
$\mathcal{M}$ accepts~$w$.
\end{lem}
\begin{proof}
For the ``if'' part, suppose $\mathcal{M}$ accepts~$w$.
Observe that $q_0\, \cent \, w \, \$ \rightarrow_R^{+} q_a\, \cent \,
w\, \$ $ by Lemma~\ref{initialAccepting}, and $R$ has the rule
$q_a \, \cent \, \rightarrow \cent$. Thus, we get
$q_0\, \cent \, w \, \$ \rightarrow_R^{+} q_a\, \cent \, w\, \$
\rightarrow \cent \, w \,\$ $.

For the ``only if'' part, suppose $\cent \,w\, \$$ is a fixed point
for $q_0$, i.e., $q_0\, \cent \, w \, \$ \rightarrow_R^{+} \cent \, w
\,\$ $. Now note that the \emph{only} rule that can remove a
state-symbol from a string is the rule $q_a \, \cent \, \rightarrow \,
\cent$. But once that rule is applied, no other rules are
applicable. Therefore, there must be a reduction sequence such that $q_0\,
\cent \, w \, \$ \rightarrow_R^{+} q_a\, \cent \, w\, \$$. This
proves that $w$~is accepted by the~DLBA $\mathcal{M}$. 
\end{proof}

\begin{thm}\label{FPConvergent}
The fixed point problem is undecidable for finite and convergent
string rewriting systems.
\end{thm}

\subsection{Fixed Point Problem for Dwindling Convergent Systems:}

Given a string $\alpha \in \Sigma^*$ and a dwindling, convergent, and
finite string rewriting system~$R$, the fixed point problem is equivalent to determining whether there exists $W \in \Sigma^*$ such that $\alpha W \rawrnorm W$. 

We shall define a shorthand notation for working with dwindling systems. Observe that applying a dwindling rule $pq \raw p$ to a
string $X\ssp{1}pqX\ssp{2}$ has the effect of "deleting" the substring $q$, and we are left with $X\ssp{1} p X\ssp{2}$. We can say that for any $X$, if $X \rawrrtc X'$ then there exists an index sequence $j_1,\dots,j_n$ such that $X' = X\strind{j_1,\dots,j_n}$, where each $j_i$ is the original position of its symbol. There may be several index sequences that satisfy $X' = X\strind{j_1,\dots,j_n}$, but for any given sequence of reductions applied to $X$ to obtain $X\strind{j_1,\dots,j_n}$, each reduction is applied at a specific position and consequently yields a particular index sequence $j_1,\dots,j_n$, derived as a subseqeunce of the previous. To denote this, we write $X \drawrrtc X\strind{j_1,\dots,j_n}$ and extend this shorthand in the obvious manner e.g. $AX \drawrrtc A\strind{j_1,\dots,j_n}X\strind{k_1,\dots,k_m}$, with the understanding that $\drawr$ is actually a relation on tuples e.g.$(X,(1,2,\dots,|X|))$,  $(X,(j_1,\dots,j_n))$, $(AX, (j_1,\dots,j_n,k_1+|A|,\dots,k_m+|A|))$.

\begin{lem}\label{dwindling_fp_l1} 
  Let $A\in\Sigma^+$ irreducible, $X\in\Sigma^+$, and $R$ dwindling,
  convergent and finite. If $AX \rightarrow_R^* Y$, then $Y = A\strind{1:m}X\strind{j_1,\dots,j_n}$ for some $m,j_1,\dots,j_n$ such that 
  \[AX \drawrrtc A\strind{1:m} X\strind{j_1,\dots,j_n}\text{.}\] Additionally, if $AX \drawrtc A\strind{1:m}X\strind{j_1,\dots,j_n}$ then $n < |X|$. 
\end{lem}

\begin{proof}
  This is easily shown by an inductive argument. Since $A$ is irreducible, no rule $pq\raw p \in R$ can match a substring of only $A$, so any sequence of reductions on $AX$ will yield a string in the form $A\strind{1:m}X\strind{j_1,\dots,j_n}$
\end{proof}

\begin{thm} \label{dwindling_fp_theorem1} If some instance $R,\alpha$ of the fixed point problem has a solution and $R$ is
  dwindling, convergent and finite then a minimal-length solution can be found in polynomial time.
\end{thm}

\begin{proof}
  The size of a minimal-length solution is bounded by $2\delta|\alpha|$, as the fixed point problem is also a special case of the conjugacy problem.\cite{narendran1985complexity} Since $R$ is convergent, it may be assumed that $\alpha$ is irreducible without loss of generality. Suppose $W$ is a minimal-length solution.

  By lemma \ref{dwindling_fp_l1}, $\alpha X \drawrnorm A\strind{1:m} W\strind{j_1,\dots,j_n}$. Ignoring the trivial case where $\alpha$ is the empty string, we have $n <|W|$.

 Since $n < |W|$ and $\alpha\strind{1:m} W\strind{j_1,\dots,j_n} = W$, it must be that $W\strind{1} = \alpha\strind{1}$. Suppose now that the first $i < |W|$ characters of $W$ are known and let $\alpha\ssp{i} := \normform{\alpha\strind{1:m\ssp{i}} W\strind{i+1:|W|}}$. Since we assume $W$ is a minimal solution, $\alpha\ssp{i}W\strind{i+1:|W|}$ must be reducible and so $\alpha\ssp{i}W\strind{i+1:|W|} \drawrnorm \alpha\ssp{i}\strind{1:m\ssp{i+1}} W\strind{j\ssp{i+1}_1,\dots,j\ssp{i+1}_{n\ssp{i+1}}}$. Since $n\ssp{i+1} < |W\strind{i+1:|W|}|$, we must have that $|\alpha\ssp{i}\strind{1:m\ssp{i+1}}| \geq i+1 $ and so we have found $W\strind{i+1}$.
\end{proof}

The algorithm above is efficient and we can readily obtain an asymptotic bound:
\begin{thm} Any instance $R,\alpha$ of a fixed point problem where $R$ is dwindling, convergent and finite can be solved in
  $O(\delta^2 |\alpha|)$ time.
\end{thm}
\begin{proof}

  Letting $\alpha\ssp{0} := \alpha$ and $\alpha\ssp{k+1} := (\alpha\ssp{k}\alpha\ssp{k}_{[k+1]})_\downarrow$, either
  $\alpha\ssp{0}_{[1]}\alpha\ssp{1}_{[2]}\dots \alpha\ssp{k-1}_{[k]}$ is a solution for $1 \leq k \leq \delta |\alpha|$ or there is
  no solution. Since $\alpha\ssp{k+1}$ is by definition irreducible, the string $\alpha\ssp{k}\alpha\ssp{k}_{[k+1]}$ is only
  reducible if the left-hand side of a rule in $R$ can be matched to a substring including the last symbol. By the nature of
  dwindling rules, reducing $\alpha\ssp{k}\alpha\ssp{k}_{[k+1]}$ to its normal form will require at most a single rewrite. The rules
  of $R$ can be arranged into a tree structure so that any matching rules can be found and applied to
  $\alpha\ssp{k}\alpha\ssp{k}_{[k+1]}$ in $O(\delta)$ time.
  \end{proof}

\begin{cor} 
If the rhs of each rule in $R$ is also of length $\leq 1$ (the cases
where $R$ is a dwindling convergent system that is also special or
monadic) and solutions exist, then there must exist a minimal-length solution that's a prefix of $\alpha$.
\end{cor} 
\begin{proof} 
Assume solutions exist and that $W$ is a solution found using the procedure in Theorem \ref{dwindling_fp_theorem1}. Let $W_{[1:k]}$ be the
smallest prefix of $W$ such that $\alpha W_{[1:k]}$ is reducible. Then $\alpha\ssp{k} := (\alpha W_{[1:k]})_\downarrow$ is a
prefix of $\alpha$, since the reduction must match a substring including $\alpha_{[|\alpha|]} W_{[1:k]}$, remove a postfix
including at least $W_{[1:k]}$ and therefore leave only a prefix $\alpha_{[1:i]}$. Inductively, if $\alpha W \drawrrtc
\alpha_{[1:i]}W_{[k+1:|W|]}$ and $k \neq W$ then there must be some minimal $\tilde{k}>k$ such that $\alpha_{[1:i]}
W_{[k+1:\tilde{k}]}$ is reducible, and again some $\tilde{i} \leq i$ such that $\alpha_{[1:i]} W_{[k+1:\tilde{k}]} \rightarrow
\alpha_{[1:\tilde{i}]}$.
\end{proof}


\section{Common Term Problem}

Note that for string rewriting systems the common term problem is
equivalent to the following problem: \medskip

\begin{description}[align=left]
  \item[Input] A string-rewriting system~$R$ on an alphabet~$\Sigma$,
and two strings~$\alpha, \beta \in \Sigma_{}^*$.
  \item[Question] Does there exist a string~$W$ such that $\alpha W \;
\stackrel{*}{{\longleftrightarrow}_R} \; \beta W$?
\end{description} \medskip

This is also known as \textit{Common Right Multiplier Problem} which has
been shown to be decidable in polynomial time for monadic and
convergent string-rewriting systems (see, e.g.,~\cite{OND98},
Lemma~3.7).  However, the CT problem is undecidable for convergent and
length-reducing string rewriting systems in
general~\cite{Otto1986}.\footnote{ In fact, Otto et al.~\cite{OND98}
showed that there is a \emph{fixed} convergent length-reducing string
rewriting system for which the CT problem is undecidable.}

In this section, we focus on the decidability of the CT problem for
\emph{convergent and dwindling} string rewriting systems. 

\subsection{Common Term Problem for Convergent and Dwindling Systems}\label{ctdwindling} We show that
the CT (Common Term) problem is undecidable for string rewriting
systems that are dwindling and convergent. We define CT as the
following decision problem: \medskip

\begin{description}[align=left]
  \item[Given] A finite, non-empty alphabet $\Sigma$, strings $\alpha,
\beta \in \Sigma^*$ and a dwindling, convergent string rewriting
system $S$.
  \item[Question] Does there exist a string $W \in \Sigma^*$ such that
$\alpha W \approx_S^{} \beta W$?
\end{description} \medskip

Note that interpreting concatenation the other way, i.e., $ab$ as
$a(b(x))$, will make this a \emph{unification} problem.

We show that the Generalized Post Correspondence Problem (GPCP) reduces
to the CT problem, where GPCP stands for a variant of the modified
post correspondence problem such that we will provide the start and
finish dominoes in the problem instance. This slight change does not
effect the decidability of the problem in any way, i.e., GPCP is
also undecidable~\cite{EKR82,GPCP}. \medskip

\begin{description}[align=left]
  \item[Given]~A finite set of tuples $\left\{(x_i,\; y_i)\right
\}^{n+1}_{i=0}$ such that each $x_i, y_i \in \Sigma^+$, i.e., for all
$i$, $|x_i|>0$, $|y_i|>0$, and $(x_0, y_0), (x_{n+1}, y_{n+1})$ are
the \emph{start} and \emph{end} dominoes, respectively.
  \item[Question]~Does there exist a sequence of indices $i_1, \ldots
,i_k$ such that \[ x_{0}\;x_{i_1}\; \ldots \;x_{i_k}\;x_{n+1} =
y_{0}\;y_{i_1}\; \ldots \; y_{i_k}\;y_{n+1} ? \]
\end{description} \medskip

\label{ExampleCT} To understand the problem described above, consider the following example. Let $C = c_0, c_1, c_2, c_3$ be a set of dominoes, where each domino is represented as a tuple of two strings. Specifically, $c_0$ is $(abb, a)$, $c_1$ is $(a, b)$, $c_2$ is $(b, a)$ and $c_3$ is $(b, abb)$. We designate $c_0$ as the start domino, and $c_3$ as the end domino. In the start domino, $c_0$, $x_0$ is equal to $abb$ and $y_0$ is equal to $a$. The problem at hand is to find a sequence of indices that generates equal strings for both elements of the tuples. Formally we seek a sequence of dominoes $c_{i_0}, c_{i_1}, \ldots, c_{i_k}$ such that resulting strings satisfy: $ x_{0}\;x_{i_1}\; \ldots \;x_{i_k}\;x_{n+1} =
y_{0}\;y_{i_1}\; \ldots \; y_{i_k}\;y_{n+1} $. In our example, a closer inspection of the dominoes reveals that the sequence of $c_0\;c_1\;c_1\;c_2\;c_3$ provides us a combination that yields equal strings: $abbaabb$ for both combinations.

We work towards showing that the CT problem defined above is
undecidable by a many-one reduction from GPCP. First, we show how to
construct a string-rewriting system that is dwindling and convergent
from a given instance of GPCP.

Let $\left\{(x_i,\; y_i)\right \}^{n}_{i=1}$ be the set of
``intermediate" dominoes and $(x_0, y_0), (x_{n+1}, y_{n+1})$, the
start and end dominoes respectively, be given. Suppose $\Sigma$ is the
alphabet given in the instance of GPCP. Without loss of generality,
we may assume $\Sigma = \{a,\, b\}$. Then the set $\hat{\Sigma} := \{a,b\}
\cup \{c_0,\;\ldots\;c_{n+1}\} \cup \{\cent_1, \cent_2, B, a_1, a_2,
a_3, b_1, b_2, b_3\} $ will be our alphabet for the instance of
CT.

Next we define a set of string homomorphisms used to simplify the
discussion of the reduction. Namely, we have the following:
\begin{equation*}
\begin{aligned}[c] h_1(a) & = & a_1 \, a_2 \, a_3,\\ h_1(b) & = &
b_1\, b_2 \, b_3,
\end{aligned} \qquad
\begin{aligned}[c] h_2(a) & = & a_1 \, a_2, \\ h_2(b) & = & b_1 \,
b_2,
\end{aligned} \qquad
\begin{aligned}[c] h_3(a) & = & a_1 \\ h_3(b) & = & b_1
\end{aligned}
\end{equation*} such that each $h_i : \Sigma \rightarrow
\hat{\Sigma}^+$ is a homomorphism.

\textbf{Reduction and form of solution:} For the convenience of the
reader, we state the form of the $CT$ problem instance and its
solutions before explaining the details. In particular, $\alpha =
\cent_1$, $\beta = \cent_2$. For instances such that a solution $Z$, which is constructed over 
the alphabet of the SRS $S$, exists, it shall take the form $h_1(Z_1)Z_2$, where $Z_1 \in \{a,b\}^+$ is
the string of the GPCP solution and $Z_2 = c_{n+1}Bc_{i_k}\dots Bc_{i_1}Bc_{0}$.
The sequence $c_{n+1},c_{i_k}, \dots, c_{i_1},c_{0}$ is the
(reversed) sequence of tuple indices, that is, the GPCP solution itself,
and they are separated by the symbol $B$. Informally, the purpose of the
homomorphisms and $\cent_1$, $\cent_2$ is to ensure that $\alpha Z$ and
$\beta Z$ are processed with two separate sets of rules, corresponding to
the sets of first and second tuple elements, respectively.

We are now in a position to construct the string rewriting system~$S$,
with the following collections of rules, named as the Class~D rules:

\begin{equation*}
\begin{aligned}[c] \cent_1 h_1(a) & \rightarrow & \cent_1 h_3(a),\\
\cent_1 h_1(b) & \rightarrow & \cent_1 h_3(b),
\end{aligned} \qquad
\begin{aligned}[c] \cent_2 h_1(a) & \rightarrow & \cent_2 h_2(a)\\
\cent_2 h_1(b) & \rightarrow & \cent_2 h_2(b)
\end{aligned}
\end{equation*}

\noindent and,
\begin{equation*}
\begin{aligned}[c] h_i(a)\,h_1(a) & \rightarrow & h_i(a)\,h_i(a),\\
h_i(b)\,h_1(a) & \rightarrow & h_i(b)\,h_i(a),
\end{aligned} \qquad
\begin{aligned}[c] h_i(a)\,h_1(b) & \rightarrow & h_i(a)\,h_i(b)\\
h_i(b)\,h_1(b) & \rightarrow & h_i(b)\,h_i(b)
\end{aligned}
\end{equation*}
\noindent for $i \in \{2, 3\}$.

The erasing rules of our system consists of three classes. Class
\rom{1} rules are defined as:
\begin{eqnarray*} \cent_1\,h_3(x_0)\,B\,c_0 & \rightarrow & \lambda \\
\cent_2\,h_2(y_0)\,c_0 & \rightarrow & \lambda
\end{eqnarray*}

and Class \rom{2} rules (for each $i = 1, 2, \ldots, n$),
\begin{eqnarray*} h_3(x_i)\,B\,c_i & \rightarrow & \lambda \\
h_2(y_i)\,c_i\,B & \rightarrow & \lambda
\end{eqnarray*}

and finally Class \rom{3} rules,
\begin{eqnarray*} h_3(x_{n+1})\,c_{n+1} & \rightarrow & \lambda \\
h_2(y_{n+1})\,c_{n+1}\,B & \rightarrow & \lambda
\end{eqnarray*}

Clearly given an instance of GPCP the above set of rules can
effectively be constructed from the instance data. Also, by
inspection, we have that our system is confluent (there are no
overlaps among the left-hand sides of rules), terminating, and
dwindling.

We then set $\alpha = \cent_1$ and $\beta = \cent_2$ to complete the
constructed instance of $CT$ from~GPCP.

It remains to show that this instance of $CT$ is a ``yes" instance if
and only if the given instance of GPCP is a ``yes" instance, i.e.,
the $CT$ has a solution if and only if the GPCP does. In that
direction, we prove some results relating to~$S$.

\begin{lem}\label{FirstStep} Suppose $\cent_1 h_3(w_1) B \gamma
\rightarrow^{!} \lambda$ and $\cent_2 h_2(w_2) \gamma \rightarrow^{!}
\lambda$ for some $w_1, w_2 \, \in \, \{a,b\}_{}^*$, then $\gamma \in
\{c_1B,\;c_2B,\;...\;,c_nB\}_{}^{*} c_0$.
  \end{lem}

  \begin{proof} Suppose $\gamma$ is a minimal counter example with
respect to length and $\gamma \in IRR(S)$. In order for the terms to
be reducible, $\gamma = c_iB\; \gamma^\prime$ (this follows by
inspection of $S$).  After we replace the $\gamma$ at the equation in
the lemma, we get:
    \begin{eqnarray*} \cent_1 \;h_3(w_1)\; B\; c_i\; B\; \gamma^\prime
& \rightarrow & \cent_1 {h_3(w_1)}^\prime B\; \gamma^\prime
\rightarrow^{!} \lambda \\ \cent_2 \;h_2(w_2)\; c_i\; B\;
\gamma^\prime & \rightarrow & \cent_2 {h_2(w_2)}^\prime \;
\gamma^\prime \rightarrow^{!} \lambda
    \end{eqnarray*}
    \noindent by applying the Class \rom{2} rules and finally Class
\rom{1} rule to erase the $\cent$ signs.  Then, however, $\gamma'$ is
also a counterexample, and $| \gamma' | < | \gamma |$, which is a
contradiction.
  \end{proof}

    We are now in a position to state and prove the main result of
this section.

  \begin{thm} The CT problem is undecidable for dwindling convergent
string-rewriting systems by a reduction from GPCP.
  \end{thm}

  \begin{proof} We first complete the ``only if" direction. Suppose CT
has a solution such that $\cent_1 Z \downarrow \cent_2 Z$ where $Z$ is
a minimal solution that is constructed over the alphabet of the string-rewriting system. We show that $Z$ corresponds to a solution for
GPCP. Let $Z_1$ be the longest string such that $h_1(Z_1)$ is a prefix
of $Z$, and denote the rest of $Z$ by $Z_2$, so that $Z$ can be written $Z = h_1(Z_1)Z_2$.
 

    $h_1(Z_1)$ can be rewritten to $h_3(Z_1)$ and $h_2(Z_1)$ by
applying the Class D rules. Thus, we will get
\begin{eqnarray*} \cent_1 \;h_1(Z_1)\; Z_2 & \rightarrow^{*} & \cent_1
h_3(Z_1)\; Z_2 \\ \cent_2 \;h_1(Z_1)\; Z_2 & \rightarrow^{*} & \cent_2
h_2(Z_1)\; Z_2
\end{eqnarray*} In order for both terms to be reducible, $Z_2$ must be of the form $Z_2 = c_{n+1}\;B\;Z_2^{\prime}$. Thus \begin{eqnarray*} \cent_1 \;h_3(Z_1)\; Z_2 & = & \cent_1 h_3(Z_1)\;
c_{n+1}\;B\;Z_2^{\prime} \\ \cent_2 \;h_2(Z_1)\; Z_2 & = & \cent_2
h_2(Z_1)\; c_{n+1}\;B\;Z_2^{\prime}
\end{eqnarray*} i.e., $Z_1 = Z_1^\prime \: x_{n+1}$ and $Z_1 =
Z_1^{\prime\prime} \: y_{n+1}$.  By applying the Class \rom{3} rules,
these equations will reduce to:
\begin{eqnarray*} \cent_1 h_3(Z_1)\; c_{n+1}\;B\;Z_2^{\prime} &
\rightarrow & \cent_1 h_3(Z_1^\prime) \; B\;Z_2^{\prime}\\ \cent_2
h_2(Z_1)\; c_{n+1}\;B\;Z_2^{\prime} & \rightarrow & \cent_2
h_2(Z_1^{\prime\prime}) \; Z_2^{\prime}
\end{eqnarray*} We now apply Lemma~\ref{FirstStep} to conclude that
$Z_2^{\prime} \in \{c_1B,\;c_2B,\; \ldots,c_nB\}_{}^{*}c_0$.\\[-5pt]

At this point we have that:
\[ Z_2 = c_{n+1}Bc_{i_k}Bc_{i_{k-1}}{\cdots}Bc_{i_2}Bc_{i_1}Bc_0 \text{\:\:\: for
some } i_1, \ldots, i_k\] Then the sequence of dominoes
\[(x_0, y_0), (x_{i_1}, y_{i_1}), {\ldots} , (x_{i_k}, y_{i_k}),
(x_{n+1}, y_{n+1}) \] will be a solution to the given instance of
GPCP with solution string~$Z_1$ since the left-hand sides of the
Class~\rom{1}, \rom{2}, \rom{3} rules consist of the images of domino
strings under $h_2$ and $h_3$. More specifically, there is a finite
number of $B$'s and $c_i$'s in~$Z_2$, so there must be a decomposition
of $h_1(Z_1)$:
\[ h_1(Z_1) = h_1(x_0)h_1(x_{i_1}) \cdots h_1(x_{i_k})h_1(x_{n+1}) \]
and
\[ h_1(Z_1) = h_1(y_0)h_1(y_{i_1}) \cdots h_1(y_{i_k})h_1(y_{n+1}) \]
Thus, we have the following reductions with Class~D rules:
\[ \cent_1\;h_1(Z_1)\;Z_2 \rightarrow^{*} \cent_1\;h_3(Z_1)\;Z_2 \]
\[ \cent_2\;h_1(Z_1)\;Z_2 \rightarrow^{*} \cent_2\;h_2(Z_1)\;Z_2 \]
Finally, by Class \rom{1}, \rom{2}, \rom{3} rules:
\[\cent_1\;h_3(Z_1)\;Z_2 \rightarrow^{*} \cent_1\; h_3(x_0)\;B\;c_0
\rightarrow \lambda \]
\[\cent_2\;h_2(Z_1)\;Z_2 \rightarrow^{*} \cent_2\; h_2(y_0)\;c_0
\rightarrow \lambda \] and $Z_1$ is a solution to the instance of
the~GPCP.

We next prove the ``if" direction. Assume that the given instance of
GPCP has a solution. Let $Z_1$ be the string corresponding to the
matching dominoes, and let
\[(x_0, y_0), (x_{i_1}, y_{i_1}), {\ldots} , (x_{i_k}, y_{i_k}),
(x_{n+1}, y_{n+1}) \] be the sequence of tiles that induces the match.
Let $Z_2 = c_{n+1}Bc_{i_k}Bc_{i_{k-1}}{\cdots}Bc_{i_2}Bc_{i_1}Bc_0$. We show that
$\cent_1 h_1(Z_1)Z_2\; \downarrow\; \cent_2 h_1(Z_1)Z_2$.  First apply the
Class $D$ rules to get:
\[ {\cent_1}h_1(Z_1)Z_2 \rightarrow^{*} {\cent_1}h_3(Z_1)Z_2 \]
\[ {\cent_2}h_1(Z_1)Z_2 \rightarrow^{*} {\cent_2}h_2(Z_1)Z_2 \] but then we
can apply Class~\rom{1}, \rom{2}, \rom{3} rules to reduce both of the
above terms to $\lambda$.
\end{proof}

This result strengthens the earlier undecidability result of
Otto~\cite{Otto1986} for string-rewriting systems that are
\emph{length-reducing} and convergent.

To clarify the results of the theorem, let us revisit the previous example \ref{ExampleCT} introduced earlier in this paper. Given the string $abbaabb$ in the example, we begin by applying homomorphism and Class D rules. This application yields the following sequence for the first string pair, $x_0\;x_1\;x_1\;x_2\;x_3$: $\cent_1 h_3(abb)h_3(a)h_3(a)h_3(b)h_3(b)c_3Bc_2Bc_1Bc_1Bc_0$ . In order to establish the reduction from GPCP to CT, we proceed to apply erasing rules to both strings. First, we employ Class \rom{3} rule to erase the last domino $h_3(b)c_3$. Subsequently, Class \rom{2} rule is applied three times to the inner dominoes: first for the homomophism $h_3(b)Bc_2$, then for $h_3(a)Bc_1$ twice. Finally, the erasing rule Class \rom{1}  removes the $\cent$ symbol along with the remaining homomorphism, resulting in: $\cent_1 h_3(abb)Bc_0 \rightarrow \lambda$. This sequence of application of the rules demonstrates the reduction process from GPCP to CT, providing a concrete illustration of the theorem's applicability.


\section{Common Equation Problem}\label{S:CE}
To clarify the CE problem for string rewriting systems, let us consider two substitutions 
$\theta^{}_1$ and $\theta^{}_2$ such that
\begin{eqnarray*}
\theta^{}_1 = \{ x_1 \mapsto \alpha^{}_1, \; x_2 \mapsto \alpha^{}_2\} \\
\theta^{}_2 = \{ x_1 \mapsto \beta^{}_1,  \; x_2 \mapsto \beta^{}_2\}
\end{eqnarray*}
Think of the letters of the alphabet as monadic function symbols as
mentioned in the notation and preliminaries section. 
We have two cases for the equation
$e_1 = e_2$: (i)~both $e_1$ and $e_2$ have the same variable in them, or
(ii)~they have different variables, i.e., one has~$x_1$ and the other~$x_2$.
Thus in the former, which we call the ``one-mapping'' case, we are looking for
\emph{different} irreducible strings~$W_1^{}$ and~$W_2^{}$ such that either
\begin{enumerate}
\item $\alpha_1 W_1 \stackrel{*}{{\longleftrightarrow}_R} \alpha_1 W_2 \; \text{ and } \;
  \beta_1 W_1 \stackrel{*}{{\longleftrightarrow}_R} \beta_1 W_2, \; \; \; \text{or}$

\item $\alpha_2 W_1 \stackrel{*}{{\longleftrightarrow}_R} \alpha_2 W_2 \; \text{ and } \;
  \beta_2 W_1 \stackrel{*}{{\longleftrightarrow}_R} \beta_2 W_2$.
\end{enumerate}
In the latter (``two-mappings case'') case, we want to find strings~$W_1^{}$ and~$W_2^{}$,
\emph{not necessarily distinct,} such that \[
\alpha_1 W_1 \stackrel{*}{{\longleftrightarrow}_R} \alpha_2 W_2 \; \text{ and } \;
 \beta_1 W_1 \stackrel{*}{{\longleftrightarrow}_R} \beta_2 W_2 \]

The loss of explicit inclusion of the variables by the passage from unary term rewriting systems to string rewriting systems is handled by the definitions of the separate problems. 

The one-mapping case can be illustrated by an example. Consider the
term rewriting system $\{ a(a(b(z))) \, \rightarrow \, a(b(z)) \}$ and
two substitutions $\theta_1^{} = \{ x \mapsto b(c) \}$ and
$\theta_2^{} = \{ x \mapsto b(b(c)) \}$. Now
$a(a(x)) \; = \; a(x)$ is a common equation. Considering this
in the string rewriting setting, we have  $R ~ = ~ \{ baa \, \rightarrow \, ba \}$,
$\alpha = b$ and $\beta = bb$. Now $W_1 ~ = ~ aa$ and $W_2 ~ = ~ a$ is a solution.\\[-4pt]

Hence we define CE as the following decision problem:

\begin{description}[align=left]
\item[Input] A string-rewriting system~$R$ on an alphabet~$\Sigma$, and 
strings~$\alpha_1, \alpha_2, \beta_1, \beta_2 \in \Sigma_{}^*$.
\item[Question] Do there exist 
\emph{irreducible} 
strings~$W_1, W_2 \in \Sigma^*$ such that one of the following conditions is satisfied?
\end{description}
\begin{tabular}{ c c r }
i. & $\alpha_1 W_1 \stackrel{*}{{\longleftrightarrow}_R} \alpha_2 W_2$, & \multirow{2}{20em}{$\;\;\;\;\;\;\;$$\alpha_1 \neq \alpha_2$ $\lor$ $\beta_1 \neq \beta_2$} \\[5pt]
& $\beta_1 W_1 \; \stackrel{*}{{\longleftrightarrow}_R} \; \beta_2 W_2$, & \\[10pt]
ii. & $\alpha_i W_1 \stackrel{*}{{\longleftrightarrow}_R} \alpha_i W_2$, & \multirow{2}{20em}{$\;\;\;\;\;\;\;$ $i \in \{1, 2\} $ $\land$ $W_1 \neq W_2$} \\[5pt]
& $ \beta_i  W_1 \; \stackrel{*}{{\longleftrightarrow}_R} \; \beta_i W_2$,  & \\[10pt]
\end{tabular}

CE is also undecidable for dwindling systems, since
in the string-rewriting case CT is a particular case of~CE. To see this,
consider the case where $\alpha_1 \, \neq \, \alpha_2$ and
$\beta_1 = \beta_2 = \lambda$, i.e., consider
the substitutions
\begin{eqnarray*}
\theta^{}_1 & = & \{ x_1 \mapsto \alpha^{}_1, \; x_2 \mapsto \alpha^{}_2\} \\
\theta^{}_2 & = & \{ x_1 \mapsto \lambda ,  \; x_2 \mapsto \lambda \}
\end{eqnarray*}
where $\alpha_1 \, \neq \, \alpha_2$.
This has a solution if and only if there are irreducible strings~$W_1, W_2 \in \Sigma^*$ such that either

\begin{tabular}{ c c r }
i. & $\alpha_1 W_1 \stackrel{*}{{\longleftrightarrow}_R} \alpha_2 W_2$, & \multirow{2}{20em}{$\;\;\;\;\;\;\;$} \\[5pt]
& $ W_1 \; \stackrel{*}{{\longleftrightarrow}_R} \; W_2$, & $\qquad$ or \\[10pt]
ii. & $\alpha_i W_1 \stackrel{*}{{\longleftrightarrow}_R} \alpha_i W_2$, & \multirow{2}{20em}{$\;\;\;\;\;\;\;$ $i \in \{1, 2\} $ $\land$ $W_1 \neq W_2$} \\[5pt]
& $   W_1 \; \stackrel{*}{{\longleftrightarrow}_R} \;  W_2$,  & \\[10pt]
\end{tabular}
\ignore{
if $\beta_1 = \beta_2 = \lambda$, then \[
\alpha_1 W_1 \; \stackrel{*}{{\longleftrightarrow}_R} \; \alpha_2 W_2 \; \; ~ \text{and} ~ \; \;
W_1 \; \stackrel{*}{{\longleftrightarrow}_R} \; W_2 \]
}

Since $W_1$ and $W_2$ are irreducible strings, $W_1 \;
\stackrel{*}{{\longleftrightarrow}_R} \; W_2$ makes $W_1$ equal to
$W_2$. Thus, we eliminate the second condition. 

With the second condition being out of the picture, we only consider the first condition which shows a similarity with the definition of Common Term(CT) problem. The CE problem for $\theta_1$ and $\theta_2$ has a solution
if and only if the CT problem for $\alpha_1$ and $\alpha_2$ has a solution. Therefore, CT is reducible to CE problem.
\ignore{
It can also be shown, using the construction
from~\cite{NarendranOtto97} that there are theories for which
CT is decidable whereas CE is not\footnote{We can use Corollary~4.4
  on page~101 of~\cite{NarendranOtto97}}. In 
  Appendix~\ref{CTdecidableCEnot}, 
  we provided a simpler construction than the one in ~\cite{NarendranOtto97} for reader's ease of understanding.
  }
We now show that CE is decidable for monadic string rewriting
systems.  We start with the \emph{two-mappings} case
first, since the solution for \emph{one-mapping} case is similar
to the \emph{two-mapping} case with a slightly simpler approach.

In the following we will need the concept of a \emph{right-factor} and \emph{minimal product (MP)}, which we define below.

\begin{defn}
Given a monadic, finite and convergent string rewriting system~$R$
and irreducible strings~$x$ and~$y$, let
$RF(x, y)$ define the set of \emph{right factors} needed to derive~$y$, i.e.,
$$RF(x, y) = \{ z \in IRR(R) ~ | ~ xz \rightarrow_R^{!} y \}.$$
\end{defn}

\begin{defn}
Let $R$ be a convergent monadic SRS. For an irreducible string~$\alpha$, let \[ MP(\alpha) 
~ = ~ \Big\{ \, w\!\big\downarrow \; ~ \Big| ~ \; w \in PREF(\alpha) \circ (\Sigma \cup \{ \lambda \})
      \Big\} \]

\noindent
\emph{MP} stands for the term \emph{Minimal Product} and \emph{PREF}
is the set of \emph{prefixes} of given string.
\end{defn}

\subsection{Two-mapping CE Problem for Monadic Systems}
For monadic and convergent string rewriting systems, the \emph{two-mappings} case of Common
Equation (CE) problem is decidable. This can be
shown using Lemma~3.6 in~\cite{OND98}. (See also
Theorem~3.11 of~\cite{OND98}.) In fact, the algorithm runs
in polynomial time as explained below:

\begin{thm}
{\label{CEFirstTheorem}}
Common Equation (CE) problem, given below, is decidable in polynomial
time for monadic,
finite and convergent string rewriting systems.
\begin{description}[align=left]
\item[Input] A string rewriting system~$R$ on an alphabet~$\Sigma$, and 
strings~$\alpha_1, \alpha_2, \beta_1, \beta_2 \in \Sigma_{}^*$.
\item[Question] Do there exist strings~$X, Y \in \Sigma^*$ such that 
$\alpha_1 X \; \stackrel{*}{{\longleftrightarrow}_R} \; \alpha_2 Y$ and 
$\beta_1 X \; \stackrel{*}{{\longleftrightarrow}_R} \; \beta_2 Y$?
\end{description}
\end{thm}

\begin{proof}
  The CE problem is a particular case of the simultaneous
  $E$-unification problem of~\cite{OND98}, but with a slight
  difference: CE consists of \emph{only two} equations, while
  simultaneous $E$-unification problem is defined for an arbitrary
  number of equations.  Besides, simultaneous $E$-unification problem is
  PSPACE-hard. We will use their construction, but we will modify
  it to obtain our polynomial time result.

Note: $RF(x, y)$ is a regular language for
all~$x$,~$y$, where $|y| \leq 1$~\cite{OND98} and a DFA
for it can be constructed in polynomial time
in $|R|$,~$|x|$ and~$|y|$. A DFA for $IRR(R)$ can also be constructed in polynomial time \cite{GIL}, \cite{Book}.  
We can characterize the solutions of the equation
$\alpha_1 X \; \downarrow_R^{} \; \alpha_2 Y$ by an
analysis similar to that used in~Lemma~3.6 of~\cite{OND98}
and its proof. 
Since $R$ is
monadic, there exist $a,b \in \Sigma \cup \{\lambda\}$ and partitions
of the strings $\alpha_1 = {\alpha_{11}} \; {\alpha_{12}}$, $\alpha_2 =
{\alpha_{21}} \; {\alpha_{22}}$, $X = {X_1}\;{X_2}$ and $Y =
{Y_1} \; {Y_2}$, such that 
\begin{tabbing}
$\qquad$ \= $\alpha_{12} X_1 \; \rightarrow_{}^! \; a$,\\
         \> $\alpha_{22} Y_1 \; \rightarrow_{}^! \; b, \qquad$ and \\[+5pt]
         \> ${\alpha_{11}}\; a\; {X_2} = {\alpha_{21}}\; b\; {Y_2}$.\\[-10pt]
\end{tabbing}
Now there are two main cases:\\[-10pt]

\begin{itemize}
\item[(a)] ${X_2}$ is a suffix of ${Y_2}$ i.e., ${Y_2} = Z\; {X_2}$ and ${\alpha_{11}}\;a = {\alpha_{21}} \;b\; Z$, therefore \[(X_1, \, Y_1 Z) ~ \in ~ RF({\alpha_{12}}, a) \times RF({\alpha_{22}}, b) \cdot Z.\]
\item[(b)] ${Y_2}$ is a proper suffix of  ${X_2}$ i.e., ${X_2} = Z''\; {Y_2}$,
  ${\alpha_{11}} \; a \; Z'' = {\alpha_{21}} b$, and
 ${\alpha_{11}} \; a \; U = {\alpha_{21}}$, $\; Z'' = U\,b$, therefore $(X_1 Z'', \, Y_1) ~ \in ~ RF({\alpha_{12}}, a) \cdot Z'' \times RF({\alpha_{22}}, b).$
\end{itemize}
\noindent
Similar partitioning can be done for the second equation. \\

Let $Sol(\alpha_1, \alpha_2)$ stand for a set of `minimal' solutions\footnote{We can define minimality in terms of
the partial order $(x, y) \, \sqsubseteq \, (xz, yz)$
for all~$z$}:

\begin{equation*}
\begin{aligned}
Sol(\alpha_1, \alpha_2) \; = & \; \; \; \; \,
 \bigcup_{\substack{a,b \; \in \; \Sigma \; \cup \; \{\lambda\} \\[+3pt] {\alpha_{11}^{} a} = {\alpha_{21}^{}}\,b\,Z}} RF({\alpha_{12}^{}}, a) \; \Cross \; RF({\alpha_{22}^{}}, b) \cdot Z \; \\[+5pt]
 & \cup 
\bigcup_{\substack{a,b \; \in \; \Sigma \; \cup \; \{\lambda\} \\[+3pt] {\alpha_{21}^{} b} = {\alpha_{11}^{}}\,a\,Z''}} RF({\alpha_{12}^{}}, a) \cdot Z'' \; \Cross \; RF({\alpha_{22}^{}}, b) \\
\end{aligned}
\end{equation*}

\noindent
Note that this is a finite union of cartesian products of regular languages. More precisely,
it is an expression of the form \[
(L_{11}^{} \Cross L_{12}^{}) \; \cup \; \ldots \; \cup \; 
(L_{m1}^{} \Cross L_{m2}^{}) \] where $m$~is a polynomial over~$| \alpha_1^{} |$,
$| \alpha_2^{} |$ and~$| \Sigma |$ and each~$L_{ij}^{}$ has a DFA of size
polynomial in~$|R|$ and~$\text{max}( |\alpha_1|, \, |\alpha_2|)$.\\
To find the complexity of a DFA concatenated with a letter or string, 
refer to the Lemma~\ref{DFAConcatLetter} for the former and Lemma~\ref{DFAConcatString} for the latter. 

\vspace{0.05in}
\noindent
The set of all solutions for  $\alpha_1^{} X \; \downarrow_R^{} \;
\alpha_2^{} Y$ is \[
\Delta(\alpha_1, \alpha_2) ~ = ~ \left\{ \vphantom{b^b} \, (w_1 x_1, \, z_1 x_1) ~ \; | \; ~ 
(w_1, z_1) \in Sol(\alpha_1, \alpha_2) \text{ and }
x_1 \in IRR(R) \right\} \]

The minimal solutions for the second
equation with $\beta_1$ and $\beta_2$, $Sol(\beta_1 , \beta_2)$, 
can be found by following the
same steps.  Thus $Sol(\beta_1 , \beta_2)$ can also be expressed
as 
the union of cartesian products of regular languages:  \[
(L_{11}^{\prime} \Cross L_{12}^{\prime}) \; \cup \; \ldots \; \cup \; 
(L_{n1}^{\prime} \Cross L_{n2}^{\prime}) \] where 
$n$ is also a polynomial over $|\beta_1|$, $|\beta_2|$ and $|\Sigma|$.
The set of all solutions for $\beta_1 X = \beta_2 Y$ equals to \[
\Delta(\beta_1, \beta_2) ~ = ~ \left\{ \vphantom{b^b} \, (w_2 x_2, \, z_2 x_2) ~ \; | \; ~ 
(w_2, z_2) \in Sol(\beta_1, \beta_2) \text{ and }
x_2 \in IRR(R) \right\} \]

The solutions for both the equations are the tuples $(w, z)
\in \; \Delta(\alpha_1^{}, \alpha_2^{})\; \cap \;\Delta(\beta_1^{},
\beta_2^{})$. That is, there must be
$w_1,\;w_2,\;z_1,\;z_2,\;x_1,\;x_2$ such that $$(w_1,\;z_1) \in
Sol(\alpha_1^{}, \alpha_2^{}), (w_2,\;z_2) \in Sol(\beta_1^{},
\beta_2^{}) \text{ and}$$ 
\begin{equation*}
\begin{aligned}
w &= w_1\; x_1= w_2 \; x_2 \\
z &= z_1\; x_1 = z_2 \; x_2
\end{aligned}
\end{equation*}
If $x_1$ is a suffix of $x_2$, i.e., $x_2 = {x_2}^{\prime} x_1$, then
\begin{equation*}
\begin{aligned}
w_1 &= w_2 \; {x_2}^{\prime}\\ 
z_1 &= z_2 \; {x_2}^{\prime}
\end{aligned}
\end{equation*}
(Similarly we repeat the same steps when $x_2$ is a suffix of~$x_1$.) 

Recall that $(w_1, \;z_1) \in L_{i1}^{} \times L_{i2}^{}$ for some
$i>0$, and  $(w_2, \;z_2) \in L_{j1}^{\prime} \times L_{j2}^{\prime}$
for some $j>0$. Thus 
$x_2^{\prime} \; \in \; L_{i1}^{} \, \backslash \, L_{j1}^{\prime}$
and $x_2^{\prime} \; \in \; L_{i2}^{} \, \backslash \, L_{j2}^{\prime}$
where $\backslash$ stands for the left quotient operation on languages,
defined as 
$A \backslash B := \{v \in \Sigma^{*} ~ | ~ \exists u \in B : uv \in A \}$ (See Lemma~\ref{FindMEFA} for more details, previous lemmas build up to Lemma~\ref{FindMEFA}).  
Thus there is a solution if
the intersection of $L_{i1}^{} \, \backslash \, L_{j1}^{\prime}$
and $L_{i2}^{} \, \backslash \, L_{j2}^{\prime}$ is nonempty. This check
has to be repeated for every~$i, \, j$. The process of finding the intersection of two languages is explained in Lemma~\ref{MEFAIntersection} 
and to be able to find the strings in the intersection and the quotient, you need to follow the steps in Lemma~\ref{FindingStringsThruIntersection}.
\end{proof}

\begin{lem}{\label{DFAConcatLetter}}
Let $\, M \; = \; (Q, \Sigma, \delta, q_0^{}, F)$ be a DFA and $a \in \Sigma$. Then
there exists a DFA\\
$\, M_{}^{\prime} \; = \; (Q_{}^{\prime}, \Sigma, \delta_{}^{\prime}, q_0^{\prime}, F_{}^{\prime})$
that recognizes $\mathcal{L}(M) \, \circ \, \{ a \}$ such that
$| F_{}^{\prime} | \leq |F|$ and $| Q_{}^{\prime} | \; \le \; | Q | + | F |$.
\end{lem}
\begin{proof}
The concatenation of a letter with a DFA can be easily achieved by
adding extra transitions from each final state $q_{fi}$ to a new
state $p_i$ for the symbol~$a$. However, that turns the DFA $M$ into a
non-deterministic finite automaton (NFA). We claim that there exists a
DFA $M_{}^{\prime}$ with $|F|$ as the upper bound for accepting states.

Consider Figure~\ref{fig:test}: $q_{f1}$ through $q_{fn}$
are the $n$ final states of the given DFA~$M$. Suppose
$\delta (q_{fi}, a) = r_i$. Then the subset construction
gives us the new transition
$\delta_{}^{\prime} (\{ q_{fi} \}, a) = \{ r_i, p_i \}$, where $\{ r_i, p_i \}$ is a label for a new state. 
The new accepting states for the
new DFA $M'$ will be $\{ r_i, p_i \}$, such that $1 \leq i \leq n$.
\begin{figure}[h]
\centering
\begin{tikzpicture}[->,>=stealth',shorten >=1pt,auto,node distance=1.8cm,
                    semithick, scale=0.5]
         \node[state, accepting] 		(F1)                   		{$q_{f1}$};
         \node[state, draw=none]	     	(Fi) [below=0.2cm of F1] 	{$\vdots$};
         \node[state, accepting]	     	(Fn) [below=0.2cm of Fi] 	{$q_{fn}$}; 
         \node[state]			     	(P1) [right of=F1] 		{$p_1$};
         \node[state]			     	(Pn) [right of=Fn] 		{$p_n$};
         \node[state] 				(R1) [below left of=F1]       {$r_1$};
         \node[state] 				(Rn) [below left of=Fn]       {$r_n$};
         \node[state]				(R2) [above left of=R1]	{$r_2$};
         \node[state]				(Rn2) [above left of=Rn]	{$r_{n2}$};
	 \node[state, draw=none]	     	(Info) [below right of=Rn]	{(a) $M$};
\path       
         (F1) edge node{$a$} (P1)
         	edge node{$a$} (R1)
	 (R1) edge node{$b$} (R2)	
         (Fn) edge node{$a$} (Pn)
         	edge node{$a$} (Rn)
	(Rn) edge node{$b$} (Rn2);
\end{tikzpicture} 
\hspace{3cm}
\begin{tikzpicture}[->,>=stealth',shorten >=1pt,auto,node distance=1.8cm,
                    semithick, scale=0.5]
         \node[state] 		(F1)                    {$q_{f1}$};
         \node[state, draw=none]	     	(Fi) [below=0.2cm of F1] {$\vdots$};
         \node[state]	     	(Fn) [below=0.2cm of Fi] {$q_{fn}$}; 
         \node[state, accepting]			     	(P1) [right of=F1] {$r_1, p_1$};
         \node[state, accepting]			     	(Pn) [right of=Fn] {$r_n, p_n$};
         \node[state] 				(R1) [below left of=F1]           {$r_1$};
         \node[state] 				(Rn) [below left of=Fn]           {$r_n$};
          \node[state]				(R2) [above left of=R1]	{$r_2$};
          \node[state]				(Rn2) [above left of=Rn]	{$r_{n2}$}; 
         \node[state, draw=none]	     	(Info)[below right of=Rn]{(b) $M'$};

\path       
         (F1) edge node{$a$} (P1)
         (Fn) edge node{$a$} (Pn)
	(P1) edge[bend right=40] node[swap]{$b$} (R2)
	(R1) edge node{$b$} (R2)
	(Pn) edge[bend left] node {$b$} (Rn2)
	(Rn) edge node{$b$} (Rn2);
\end{tikzpicture}
\caption{DFA $M$ concatenation with a single letter $a$.}
\label{fig:test}
\end{figure}
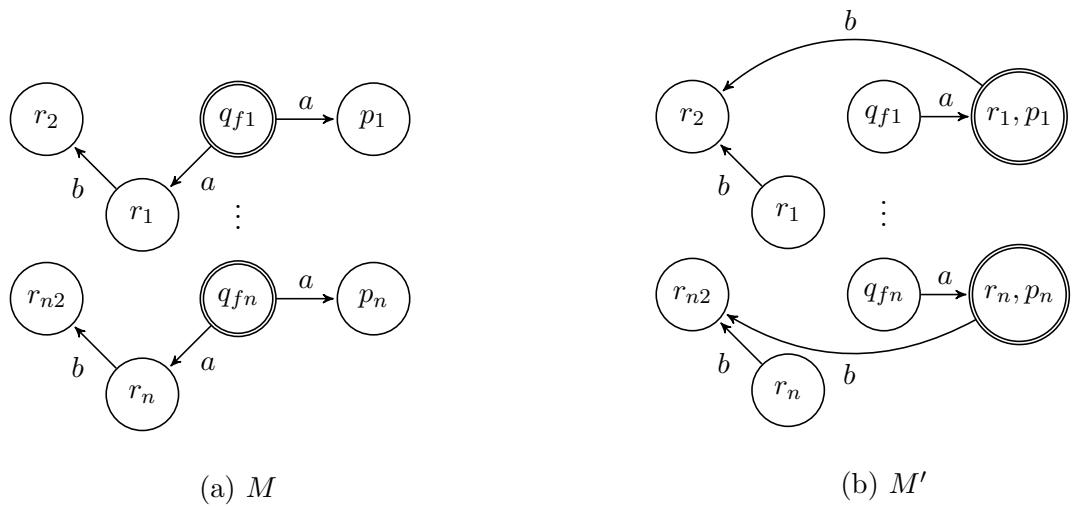

\ignore{
However, if the letter we want to concatenate does not have a
transition from the final states of $M$, the state of the string, $p_1
\ldots p_n$ will be the new final states of the DFA $M_{}^{\prime}$.
}

Besides, if the transitions for the letter~$a$ from
two earlier accepting states have the same destination state,
we can combine the new accepting states that were created.
Thus in
Figure~\ref{fig:sameLetterTransition}, the state
$\{ r_1, p_1 \}$ can be assigned to~$\delta^{\prime} (\{q_{fn}\}, a)$,
avoiding needless duplication.

Thus the number of final states, $|F_{}^{\prime}|$
for the DFA $M_{}^{\prime}$ is less than or equal to the original
number of final states, $|F|$, in DFA~$M$.

Total number of states $|Q_{}^{\prime}|$ for $M_{}^{\prime}$ is
bounded by the number of the final states $|F|$ in~$M$ as well as the
number of total states, $|Q|$, in~$M$.  Therefore, the
number of states for~$M_{}^{\prime}$ can be less than or equal to the
both of the factors, i.e., $|Q_{}^{\prime}| \leq |Q|+|F|$.
\end{proof}

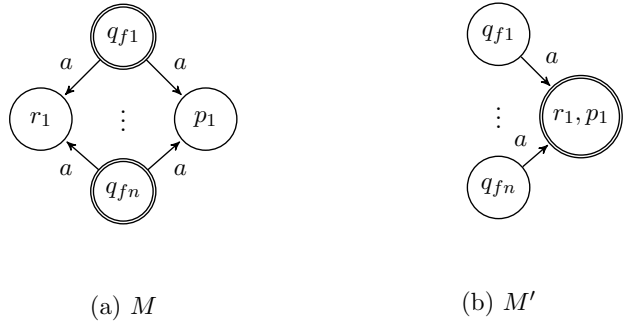
\begin{figure}
\centering
\resizebox{.55\linewidth}{!}{
\begin{tikzpicture}[->,>=stealth',shorten >=1pt,auto,node distance=1.8cm,
                    semithick, scale=0.5]
         \node[state, accepting] 		(F1)                   		{$q_{f1}$};
         \node[state, draw=none]	     	(Fi) [below=0.2cm of F1] 	{$\vdots$};
         \node[state, accepting]	     	(Fn) [below=0.2cm of Fi] 	{$q_{fn}$}; 
         \node[state]			     	(P1) [below right of=F1] 		{$p_1$};
         \node[state] 				(R1) [below left of=F1]       	{$r_1$};
	 \node[state, draw=none]	     	(Info) [below of=Fn]	{(a) $M$};
\path       
         (F1) edge node{$a$} (P1)
         	edge node[swap]{$a$} (R1)	
         (Fn) edge node[swap]{$a$} (P1)
         	edge node{$a$} (R1);
\end{tikzpicture} 
\hspace{3cm}
\begin{tikzpicture}[->,>=stealth',shorten >=1pt,auto,node distance=1.8cm,
                    semithick, scale=0.5]
         \node[state] 				(F1)                    		{$q_{f1}$};
         \node[state, draw=none]	     	(Fi) [below=0.2cm of F1] 	{$\vdots$};
         \node[state]	     			(Fn) [below=0.2cm of Fi] 	{$q_{fn}$}; 
         \node[state, accepting]		(P1) [below right of=F1] 		{$r_1, p_1$};
         \node[state, draw=none]	     	(Info)[below of=Fn]	{(b) $M_{}^{\prime}$};

\path       
         (F1) edge node{$a$} (P1)
         (Fn) edge node{$a$} (P1);
\end{tikzpicture}
}
\caption{DFA $M_{}^{\prime}$ can have less than $|F|$ states.}
\label{fig:sameLetterTransition}
\end{figure}

\begin{lem}{\label{DFAConcatString}}
Concatenation of a deterministic finite automaton (DFA) with a single
string has the time complexity $O(|F| \ast |Z| \ast
|\Sigma|)$, where $|F|$ is the number of final states in the DFA, $|Z|$
is the length of the string and $|\Sigma|$ is the size of
the alphabet.
\end{lem}
\begin{proof}
Recall that the previous lemma proved that the number of states 
in the new DFA after
the concatenation of one letter 
is at most~$|Q| + |F|$ and that the number of 
(new) accepting states is at most~$F$. Thus repeatedly
applying this operation will result in a DFA with
at most~$|Q| + |Z| \ast |F|$ states and at most~$|F|$ accepting states.
The number of new edges will be at most~$|F| \ast |Z| \ast
|\Sigma|$. Thus the overall complexity is
polynomial in the size of the original~DFA. The Figure~\ref{fig:cloudDFA} shows an example of DFA $M$'s concatenation with the symbols in string $Z$. 
\end{proof}

\begin{figure}
\begin{tikzpicture}[->,>=stealth',shorten >=1pt,auto,node distance=1.8cm,
                    semithick]
  \coordinate (c) at (2,0);
  \draw[black,ultra thick,rounded corners=1mm] (c) \irregularcircle{3cm}{2mm};
  \node[initial,state, circle] (A)                    {$q_0$};
  	
  \node[state, draw=none]         (H) [right of=A] {$\hdots$};
  \node[state, draw=none]	     (Fi) [right of=H] {$\vdots$};
  \node[state, accepting]	     (F1) [above of=Fi] {$q_{a1}$};
  \node[state, accepting]	     (Fn) [below of=Fi] {$q_{an}$}; 
  \node[state]			     (P1) [right =2.1 cm of Fi] {$p_1$}; 
  \node[state]			     (P2) [right of=P1] {$p_2$};
  \node[state, draw=none]	     (Pi) [right of=P2] {$\hdots$};
  \node[state]			     (Pn) [right of=Pi] {$p_n$};
  
\path (A) edge (H)
	(H) edge [bend left] (F1)
	      edge (Fi)
	      edge [bend right] (Fn)
	(F1) edge [bend left, blue] node{$a_1$} (P1)
	(Fi)  edge [blue] node{$a_1$} (P1)
	(Fn)  edge [bend right, blue] node{$a_1$} (P1)
	(P1) edge node{$a_2$} (P2)
	(P2) edge node{$a_3$} (Pi)
	(Pi) edge node{$a_n$} (Pn);
  
\end{tikzpicture}
\caption{DFA $M$ concatenation with a string $Z = a_1 a_2 \ldots a_n$.}
\label{fig:cloudDFA}
\end{figure}
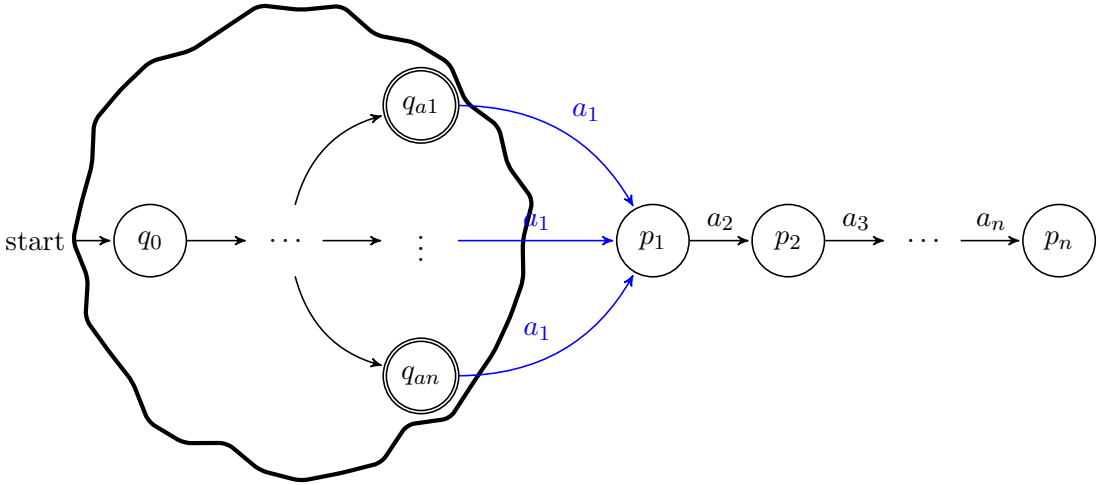

\begin{lem}\label{FindMEFA}
Let $M_1^{}$ and $M_2^{}$ be DFAs. Then a multiple-entry DFA (MEFA)
\\for $\mathcal{L}(M_2^{}) \, \backslash \, \mathcal{L}(M_1^{})$ can
be computed in polynomial time.
\end{lem}

\begin{proof}
Let $M_1^{} = (Q_1^{}, \Sigma, \delta_1^{}, q_{01}^{}, F_1^{})$ and
$M_2^{} = (Q_2^{}, \Sigma, \delta_2^{}, q_{02}^{}, F_2^{})$,
where $\delta_i : Q_i \times \Sigma \rightarrow Q_i$ are the transition functions 
and $\delta_i^*$ denotes their extensions to $Q_i \times \Sigma^*$. 

Let $L_{quo} = \mathcal{L}(M_2) \backslash \mathcal{L}(M_1)$. Then 
\begin{align*}
    y \in L_{quo} &\Leftrightarrow \exists x \in \mathcal{L}(M_2) \text{ such that } xy \in \mathcal{L}(M_1) \text{ by definition of left quotient} \\
    & \Leftrightarrow \delta_2^*(q_{02}, x) = q \in F_2 \text{ and } \delta_1^*(xy, q_{01}) = p' \in F_1 \\
    & \Leftrightarrow  \delta_2^*(q_{02}, x) = q \in F_2 \text{ and } \delta_1^*(x, q_{01}) = p \in Q_1 \text{ and } \delta_1^*(y, p) = p'
\end{align*}

Let $P$ be the product transition system of the two automata, i.e., $$P
\; = \; (Q_1^{} \times Q_2^{}, \Sigma, \delta, (q_{01}^{} ,
q_{02}^{}))$$ where $\delta$ is defined as \[ \delta((r, r'), c) = (
\delta_1^{}(r, c) , \, \delta_2^{}(r', c) ) \] for all~$c \in \Sigma$,
$r \in Q_1^{}$ and $r' \in Q_2^{}$. We can assume that states which
are not reachable from~$(q_{01}^{} , q_{02}^{})$ have been removed.
This can be done in time linear in the size of the transition graph.

Now in terms of the product transition system we can say that a string~$y$
belongs to~$L_{quo}^{}$ if and only if there exists a string~$x$ and
states $(p,q) \in Q_1 \times F_2$ and  $(p',q') \in F_1 \times Q_2$ such that 
\[ \delta_{}^* ((q_{01}^{} , q_{02}^{}), \, x) ~ = ~(p, q) \text{  and  } \delta_{}^*((p, q), y) ~ = ~ (p', q') \] 
We can now convert $P$ into a multiple-entry DFA~(MEFA). In the above case, $(p,q)$ has to be
one of the initial states of the new MEFA, and $(p', q')$ one
of its final states. Therefore, the states that are reachable
from $(q_{01}^{} , q_{02}^{})$ that are in $F_2^{} \times Q_1^{}$ will be
the initial states of the MEFA and $F_1^{} \times Q_2^{}$ will be the
final states of the~MEFA.
\end{proof}

\begin{lem}{\label{MEFAIntersection}}
Given two MEFAs, we can check whether their intersection is empty in polynomial time.
\end{lem}

\begin{proof}
Consider two MEFAs $A_1^{} \; = \; (Q_1^{}, \Sigma, \delta_1^{},
Q_{s_1}^{}, F_{A_1}^{})$ and $A_2^{} \; = \; (Q_2^{}, \Sigma,
\delta_2^{}, Q_{s_2}^{}, F_{A_2}^{})$. $Q_{s_1}^{}$ and
$Q_{s_2}^{}$ may include more than one initial state.
A string~$w$ is accepted by both MEFAs if and only if
there exist states 
$q_{init}^1 \; \in \; Q_{s_1}^{}$ and
$q_{init}^2 \; \in \; Q_{s_2}^{}$ such that
$$\delta_1^*( q_{init}^1 , w ) \; \in \; F_{A1}^{} \text{ and }
\delta_2^*( q_{init}^2 , w ) \; \in \; F_{A2}^{}.$$

To find such a string~$w$, we take the product transition system of the
two MEFAs, named as $T$, i.e., $T \; = \; (Q_1^{} \times Q_2^{},
\Sigma, \delta, (Q_{s_1}^{} \times Q_{s_2}^{}))$ where
\[ \delta((r, r'), c) = (
\delta_1^{}(r, c) , \, \delta_2^{}(r', c) ) \] for all~$c \in \Sigma$,
$r \in Q_1^{}$ and $r' \in Q_2^{}$.  A
string~$w$ is accepted by both of the MEFAs $A_1^{}$ and $A_2^{}$
if and only if there exist states $p, q, p', q'$ such that
$p \; \in \; Q_{s_1}^{}$, $q \; \in \; Q_{s_2}^{}$,
$p' \; \in \; F_{A_1}$ and $q' \; \in \; F_{A_2}$, and \[ \delta_{}^*((p , q) , w) ~ = ~ (p', q')
\] We can now apply \emph{depth-first search (DFS)} to check, in time
linear in the size of~$T$, 
if there exists a path from some state in $Q_{s_1}^{} \times Q_{s_2}^{}$ to a state 
in~$F_{A_1} \times F_{A_2}$.
\end{proof}
\vspace{0.1cm}
\begin{lem}{\label{FindingStringsThruIntersection}}
The following problem is decidable in polynomial time:

  \begin{quote}
    \begin{quote}
  \begin{description}[align=right]
  \item[{\bf Input}] DFAs $M^{}$ and $N^{}$.\\
  \item[{\bf Question}] Do there exist strings $x, y, z$ such that $x \neq y$,
    $x, y \in \mathcal{L}( M^{} )$, \emph{and}
    \mbox{$xz, yz \in \mathcal{L}( N^{} )$}?
  \end{description}
    \end{quote}
  \end{quote}
\end{lem}

\begin{proof}
Suppose there exist strings $x, y, z$ such that $x \neq y$, $x, y \in
\mathcal{L}( M^{} )$, \emph{and} \mbox{$xz, yz \in \mathcal{L}( N^{}
  )$}. We call the triple $(x, y, z)$ a \emph{solution}. Thus, we have two cases:

\begin{enumerate}[label=(\roman*)]
\item Both $x$ and $y$ start from an initial state $q_0$ and reach the
  same state, $q$, in $N$, i.e., \[ \exists\; q: \delta_{}^*(q_0^{N},
  x) = \delta_{}^*(q_0^{N}, y) = q \; \; \mathrm{and} \; \; \delta_{}^*(q, z) \in
  F^{N}. \]
\item $x$ and $y$ reach different states, say $q^{\prime}$~and~$q^{\prime\prime}$,
  in~$N$, i.e., \[ \exists\; q^{\prime}, q^{\prime\prime}: \delta_{}^*(q_0^{N},
  x) = q^{\prime} \neq q^{\prime\prime} = \delta_{}^*(q_0^{N}, y) \; \; \mathrm{and} \; \;
  \delta_{}^*(q^{\prime}, z) \in F^{N} \land \delta_{}^*(q^{\prime\prime}, z) \in
  F^{N}. \]
\end{enumerate}

Let $A=(Q, \Sigma, \delta, s, F)$ be a DFA and $p$ be a state in
$A$. By $A^{F = \{ p \}}$, we denote a replication of $A$, with
the sole difference of $p$ being the only accepting state of
$A$. Thus $N^{F = \{ q \}}$ denotes a replication of
$N$, with $q$ being the accepting state of $N$. Then, we classify
these states of~$N$ which \emph{are not dead states} into GREEN,
ORANGE and BLUE states. Note that confirming the status of~$q$
being a dead state can be done in linear time w.r.t.\ to the size
of graph.
\begin{itemize}
\item \emph{GREEN states:} $\big\{ q \mid \; \scalebox{1.2}{$\mid$}  \mathcal{L}( N^{F=\{ q \}}) \cap \mathcal{L}( M^{}) \scalebox{1.2}{$\mid$} > 1 \big\}$. \\
  
  The state~$q$ mentioned in case~(i) is a GREEN state. (See Figure~\ref{fig:sub1})

\item \emph{ORANGE states:} $\big\{ q \mid \; \scalebox{1.2}{$\mid$}  \mathcal{L}( N^{F=\{ q \}}) \cap \mathcal{L}( M^{}) \scalebox{1.2}{$\mid$} = 1 \big\}$. \\
  
Suppose that case (i) does not apply, i.e., there are no GREEN states in~$N$. 
Then case (ii) must apply and the states $q^{\prime}$
and $q^{\prime\prime}$ must be ORANGE states; in other words, the intersection of
$\mathcal{L}( M^{})$ individually with the two DFAs, $\mathcal{L}(
N^{F=\{ q^{\prime} \}})$ and $\mathcal{L}( N^{F=\{ q^{\prime\prime} \}})$ gives us
exactly 1 string for each. Note also that $x$ and $y$ are two
strings in~$\mathcal{L}( M^{})$ which are not equal to each other since
$q_{}^{\prime} \; \neq \; q_{}^{\prime\prime}$. (See Figure~\ref{fig:sub2})

\item \emph{BLUE states:} $\big\{ q \mid \; \scalebox{1.2}{$\mid$}
  \mathcal{L}( N^{F=\{ q \}}) \cap \mathcal{L}( M^{})
  \scalebox{1.2}{$\mid$} = 0 \big\}$. \\
\end{itemize}


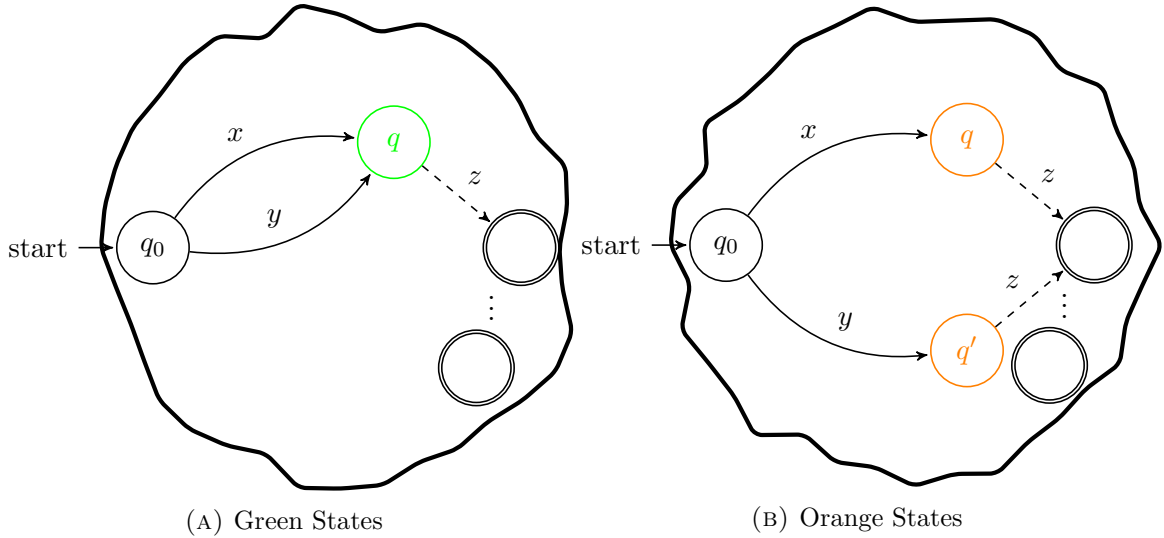
\begin{figure}[h]
\centering
\begin{subfigure}{.5\textwidth}
\centering
\begin{tikzpicture}[->,>=stealth',shorten >=1pt,auto,node distance=1.8cm,
                    semithick]
  \coordinate (c) at (2.5,0);
  \draw[black,ultra thick,rounded corners=1mm] (c) \irregularcircle{3.1cm}{2mm};
  \node[initial,state, circle] (A)                    {$q_0$};
  	
  \node[state, draw=none]	     (Fi) [right of=A] {};
  \node[state, green]	     (F1) [above right=1 cm of Fi] {$q$};
  \node[state, accepting]			     (P1) [right=2.1 cm of Fi] {}; 
  \node[state, accepting]	     (P2) at (4.3,-1.60) {} ; 
  \node[state, draw=none]	     (D) at (4.5, -0.7) {$\vdots$};
  
\path 
	(A) edge [bend left] node{$x$} (F1)
	      edge [bend right] node{$y$} (F1)
	(F1) edge[dashed] node{$z$} (P1);
  
\end{tikzpicture}
\caption{Green States}
  \label{fig:sub1}
\end{subfigure}%
\begin{subfigure}{.5\textwidth}
\centering
\begin{tikzpicture}[->,>=stealth',shorten >=1pt,auto,node distance=1.8cm,
                    semithick]
  \coordinate (c) at (2.5,0);
  \draw[black,ultra thick,rounded corners=1mm] (c) \irregularcircle{3.1cm}{2mm};
  \node[initial,state, circle] (A)                    {$q_0$};
  	
  \node[state, draw=none]	     (Fi) [right of=A] {};
  \node[state, orange]	     (F1) [above right=1 cm of Fi] {$q$};
  \node[state, orange]	     (Fn) [below right=1 cm of Fi] {$q^{\prime}$}; 
  \node[state, accepting]			     (P1) [right=2.1 cm of Fi] {}; 
  \node[state, accepting]	     (P2) at (4.3,-1.60) {} ; 
  \node[state, draw=none]	     (D) at (4.5, -0.7) {$\vdots$};
  
\path 
	(A) edge [bend left] node{$x$} (F1)
	      edge [bend right] node{$y$} (Fn)
	(F1) edge[dashed] node{$z$} (P1)
	(Fn)  edge [dashed] node{$z$} (P1);
  
\end{tikzpicture}
\caption{Orange States}
  \label{fig:sub2}
\end{subfigure}
\caption{Different Structures of Green and Orange States.}
\label{fig:GOStates}
\end{figure}

The algorithm for finding the triple $(x, y, z)$ is constructed
as follows. First, we identify the green and orange states. If
there exists a green state, then we have a solution. Otherwise we
explore whether there exists a $z$ such that $\delta_{}^*(q^{\prime},
z) \in F^{N} \land \delta_{}^*(q^{\prime\prime}, z) \in F^{N}$
for orange states~$q_{}^{\prime}$
and~$q_{}^{\prime\prime}$, i.e., we check whether \[ \big\{ \, z \; \big| \; \exists (q^{\prime}, q^{\prime\prime}) : \, q^{\prime}, q^{\prime\prime} \text{ are
  orange states } \land \; \delta_{}^*(q^{\prime},z) \in F^{N} \; \land \; 
  \delta_{}^*(q^{\prime\prime},z) \in F^{N} \big\} \] is empty.

Given orange states
$q^{\prime}$ and $q^{\prime\prime}$, we use DFA intersection to check whether there is
a string~$z$ that takes both to an accepting state.
Let $N_{s=\{q^{\prime}\}}$ denote a
replication of $N$, with the difference of $q^{\prime}$ being the initial state
of~$N$. $N_{s=\{q^{\prime\prime}\}}$ is similar to the $N_{s=\{q^{\prime}\}}$, but
this time $q^{\prime\prime}$ is the initial state. After creating these two
DFAs, we can find if there exists a string~$z$ by intersecting the
DFA~$N_{s=\{q^{\prime}\}}$ with~$N_{s=\{q^{\prime\prime}\}}$. 
This process may have to be repeated for every tuple $(q^{\prime},
q^{\prime\prime})$ of orange states. \qedhere
\end{proof}

  \ignore{
        Given $x, y \in \mathcal{L}( M_1^{} )$, \emph{and} \mbox{$xz, yz \in \mathcal{L}( M_2^{} )$}, we can extract $z$ by taking the quotient of languages $\mathcal{L}( M_1^{} )$ and $\mathcal{L}( M_2^{} )$ , such that $z \in 		\mathcal{L}( M_2^{} ) \backslash \mathcal{L}( M_1^{} )$.
        
        The quotient of the languages $\mathcal{L}( M_1^{} )$ and $\mathcal{L}( M_2^{} )$ can be find by using \emph{Lemma~\ref{FindMEFA}}. Pay attention to the fact that the quotient of the languages creates a \emph{MEFA}, 		which has more than 1 initial state.
        
        For such a $z$, there exists different cases to discover $x$ and $y$'s, such that $x \neq y$:
        \begin{enumerate}[label=\roman*.]
        \item $\exists x_1,\; x_2 \in \mathcal{L}( M_1^{} )$ such that $x_1 z \in \mathcal{L}( M_2^{} )$ and $x_2 z \in \mathcal{L}( M_2^{} )$. 
        
        Therefore, $| \mathcal{L}( M_2^{} ) \cap \mathcal{L}( M_1^{} ) \circ \{ z\} | > 1$. \ignore{cdot}
        \item $\exists y_1,\; y_2 \in \mathcal{L}( M_1^{} )$ such that $y_1 z \in \mathcal{L}( M_2^{} )$ and $y_2 z \in \mathcal{L}( M_2^{} )$. 
        
        Same construction with the previous case, with the difference of $x$ changed with $y$. Preserves the result of $| \mathcal{L}( M_2^{} ) \cap \mathcal{L}( M_1^{} ) \circ \{ z\} | > 1$. 
        \item ...
        \end{enumerate}
        }


\begin{lem}{\label{CEOneMappingMonadicLemma1}}
Let $\mu, \omega, X, Y \in IRR(R)$. Then $\mu X  \downarrow  \omega Y$ if and only
if there exist strings~$X^{\prime}, Y^{\prime}, W, \gamma$ such that
\begin{enumerate}

\item $\gamma \in MP(\mu) \cup MP(\omega)$,

\item $X = X^{\prime} W$, $Y = Y^{\prime} W$, \emph{and}

\item $\mu X^{\prime} \stackrel{!}{{\longrightarrow}_{R}^{}} ~ \gamma \; \; ~ {}_R^{}{\!{\stackrel{!}{\longleftarrow}}} \; \omega Y^{\prime}$.

\end{enumerate}


\end{lem}
\begin{proof}
This proof follows from~\cite{OND98}
(see Lemma~3.6).
\end{proof}


\begin{lem}{\label{CEOneMappingMonadicLemma2}}
\openup 0.5em
Let $\alpha_1^{}, \alpha_2^{}, \beta_1^{} , \beta_2^{} , X, Y \in IRR(R)$. Then $\alpha_1^{} X  \downarrow  \alpha_2^{} Y$
and $\beta_1^{} X  \downarrow  \beta_2^{} Y$ if and only
if there exist strings~$X^{\prime}, Y^{\prime}, V, W, \gamma_1^{}, \gamma_2^{}$ such that

\openup -0.5em
\begin{enumerate}

\item $\gamma_1^{} \in MP(\alpha_1^{}) \cup MP(\alpha_2^{})$,

\item $\gamma_2^{} \in MP(\beta_1^{}) \cup MP(\beta_2^{})$,

\item $X = X^{\prime} V W$, $Y = Y^{\prime} V W, \; \; $ \emph{and}

\item either \begin{itemize}

  \item[(a)] $\alpha_1^{} X^{\prime} V \stackrel{!}{{\longrightarrow}_{R}^{}} ~ \gamma_1^{} \; \; ~ {}_R^{}{\!{\stackrel{!}{\longleftarrow}}} \;
  \alpha_2^{} Y^{\prime} V \; \; $ \emph{and}

  \item[(b)] $\; \; \; \beta_1^{} X^{\prime} \stackrel{!}{{\longrightarrow}_{R}^{}} ~ \gamma_2^{} \; \; ~ {}_R^{}{\!{\stackrel{!}{\longleftarrow}}} \;
  \beta_2^{} Y^{\prime}$.
  
\end{itemize} or \begin{itemize}

  \item[(c)] $\; \; \; \alpha_1^{} X^{\prime} \stackrel{!}{{\longrightarrow}_{R}^{}} ~ \gamma_1^{} \; \; ~ {}_R^{}{\!{\stackrel{!}{\longleftarrow}}} \;
  \alpha_2^{} Y^{\prime} \; \; $ \emph{and}

  \item[(d)] $\beta_1^{} X^{\prime} V \stackrel{!}{{\longrightarrow}_{R}^{}} ~ \gamma_2^{} \; \; ~ {}_R^{}{\!{\stackrel{!}{\longleftarrow}}} \;
  \beta_2^{} Y^{\prime} V$.
  \end{itemize}

\end{enumerate}

\begin{proof}
This proof also follows from~\cite{OND98}
(see Lemma~3.6).
\end{proof}


\end{lem}

\vspace{0.2in}

\begin{cor}{\label{CorB3Solvability}}
\openup 0.5em
  Let $\alpha_1^{}, \alpha_2^{}, \beta_1^{} , \beta_2^{} \in
  IRR(R)$. Then there exist irreducible strings~$X$~and~$Y$ such that
  $\alpha_1^{} X \downarrow \alpha_2^{} Y$ and $\beta_1^{} X
  \downarrow \beta_2^{} Y$ if and only if there exist
  strings~$X^{\prime}, Y^{\prime}, V, \gamma_1^{}, \gamma_2^{}$
  such that

\openup -0.5em
\begin{enumerate}

\item $\gamma_1^{} \in MP(\alpha_1^{}) \cup MP(\alpha_2^{})$,

\item $\gamma_2^{} \in MP(\beta_1^{}) \cup MP(\beta_2^{})$,

\item $X^{\prime} V$ and $Y^{\prime} V$ are irreducible, \emph{and} 

\item either \begin{itemize}

  \item[(a)] $\alpha_1^{} X^{\prime} V \stackrel{!}{{\longrightarrow}_{R}^{}} ~ \gamma_1^{} \; \; ~ {}_R^{}{\!{\stackrel{!}{\longleftarrow}}} \;
  \alpha_2^{} Y^{\prime} V \; \; $ \emph{and}

  \item[(b)] $\; \; \; \beta_1^{} X^{\prime} \stackrel{!}{{\longrightarrow}_{R}^{}} ~ \gamma_2^{} \; \; ~ {}_R^{}{\!{\stackrel{!}{\longleftarrow}}} \;
  \beta_2^{} Y^{\prime}$.
  
\end{itemize} or \begin{itemize}

  \item[(c)] $\; \; \; \alpha_1^{} X^{\prime} \stackrel{!}{{\longrightarrow}_{R}^{}} ~ \gamma_1^{} \; \; ~ {}_R^{}{\!{\stackrel{!}{\longleftarrow}}} \;
  \alpha_2^{} Y^{\prime} \; \; $ \emph{and}

  \item[(d)] $\beta_1^{} X^{\prime} V \stackrel{!}{{\longrightarrow}_{R}^{}} ~ \gamma_2^{} \; \; ~ {}_R^{}{\!{\stackrel{!}{\longleftarrow}}} \;
  \beta_2^{} Y^{\prime} V$.\\
  \end{itemize}

\end{enumerate}
\end{cor}

\begin{proof}
The ``if'' part is obvious since
  $\alpha_1^{} X^{\prime} V \downarrow \alpha_2^{} Y^{\prime} V$ and $\beta_1^{} X^{\prime} V
  \downarrow \beta_2^{} Y^{\prime} V$ in both cases~(a) and~(b). The ``only if'' part follows
  from~Lemma~\ref{CEOneMappingMonadicLemma2}.
\end{proof}

\noindent
Recall the definition 
$RF(x, y) = \{ z \in IRR(R) ~ | ~ xz \rightarrow_R^{!} y \}.$
Thus Corollary~\ref{CorB3Solvability} can be restated as

\begin{lem}{\label{CEOneMappingMonadicLemma4}}
\openup 0.5em
  Let $\alpha_1^{}, \alpha_2^{}, \beta_1^{} , \beta_2^{} \in
  IRR(R)$. Then there exist irreducible strings~$X$~and~$Y$ such that
  $\alpha_1^{} X \downarrow \alpha_2^{} Y$ and $\beta_1^{} X
  \downarrow \beta_2^{} Y$ if and only if there exist
  strings~$X^{\prime}, Y^{\prime}, V, \gamma_1^{}, \gamma_2^{}$
  such that

\openup -0.5em
\begin{enumerate}

\item $\gamma_1^{} \in MP(\alpha_1^{}) \cup MP(\alpha_2^{})$,

\item $\gamma_2^{} \in MP(\beta_1^{}) \cup MP(\beta_2^{})$,

\item either\\[-10pt] \begin{itemize}
  
  \item[(a)] $X^{\prime} V \, \in \, RF(\alpha_1, \gamma_1^{}) $, $~ Y^{\prime} V \, \in \, RF(\alpha_2, \gamma_1^{}) $, 
             $~ X^{\prime} \, \in \, RF(\beta_1, \gamma_2^{})$ \emph{and} $~ Y^{\prime} \, \in \, RF(\beta_2, \gamma_2^{})$
  
\end{itemize} or \begin{itemize}

  \item[(b)] $X^{\prime} \, \in \, RF(\alpha_1, \gamma_1^{}) $, $~ Y^{\prime} \, \in \, RF(\alpha_2, \gamma_1^{}) $, 
             $~ X^{\prime} V  \, \in \, RF(\beta_1, \gamma_2^{})$ \emph{and} $~ Y^{\prime} V \, \in \, RF(\beta_2, \gamma_2^{})$

\end{itemize}

\end{enumerate}

\end{lem}

\noindent
Note that the existence of~$X^{\prime}$, $Y^{\prime}$ and~$V$ in 
statements~(a) and~(b) in the above lemma can be formulated in terms of the
left quotient operation. Thus 

\begin{thm}{\label{CETwoMappingsTheorem}}
\openup 0.5em
  Let $\alpha_1^{}, \alpha_2^{}, \beta_1^{} , \beta_2^{} \in
  IRR(R)$. Then there exist irreducible strings~$X$~and~$Y$ such that
  $\alpha_1^{} X \downarrow \alpha_2^{} Y$ and $\beta_1^{} X
  \downarrow \beta_2^{} Y$ if and only if there exist
  strings~$\gamma_1^{}, \gamma_2^{}$
  such that

\openup -0.5em
\begin{enumerate}

\item $\gamma_1^{} \in MP(\alpha_1^{}) \cup MP(\alpha_2^{})$,

\item $\gamma_2^{} \in MP(\beta_1^{}) \cup MP(\beta_2^{})$,

\item either\\[-10pt] \begin{itemize}
  
\item[(a)] $RF(\beta_1, \gamma_2^{}) \, \backslash \, RF(\alpha_1, \gamma_1^{}) ~ \; \cap \; ~
            RF(\beta_2, \gamma_2^{}) \, \backslash \, RF(\alpha_2, \gamma_1^{}) ~ ~ \neq ~ ~ \emptyset$
  
\end{itemize} or \begin{itemize}

\item[(b)] $RF(\alpha_1, \gamma_1^{}) \, \backslash \, RF(\beta_1, \gamma_2^{}) ~ \; \cap \; ~
            RF(\alpha_2, \gamma_1^{}) \, \backslash \, RF(\beta_2, \gamma_2^{}) ~ ~ \neq ~ ~ \emptyset$

\end{itemize}

\end{enumerate}

\end{thm}

\begin{proof}
The proof follows from the above lemmas and corollaries.
\end{proof}

\subsection{One-mapping CE Problem for Monadic Systems}
The one-mapping case of the CE problem is decidable for monadic and
convergent string rewriting systems.  We can show it using a
construction similar to the two-mappings case. However it will be a
slightly simpler approach since we only have two input
strings~$\alpha$ and~$\beta$ as opposed to four.  The algorithm
for the one-mapping case also runs in polynomial time as explained below:

\begin{thm}{\label{CESecondTheorem}}
The following CE problem is decidable in polynomial time for monadic, finite and convergent string rewriting systems.
\begin{description}[align=left]
\item[Input] A string-rewriting system~$R$ on an alphabet~$\Sigma$, and irreducible
strings~$\alpha, \beta \in \Sigma_{}^*$.
\item[Question] Do there exist distinct irreducible strings~$X, Y \in IRR(R)$ such that 
\begin{enumerate}
\item[] $\alpha X \downarrow_R^{} \alpha Y \; \text{ and } \;
  \beta X \downarrow_R^{} \beta Y$?
\end{enumerate}
\end{description}
\end{thm}
\begin{proof}
This follows from Lemma~\ref{CEOneMappingMonadicLemma4} and \ref{unrestricted_rf_regular},
since we only need to check whether there are
strings $X^{\prime}, Y^{\prime}, V, \gamma_1^{}, \gamma_2^{}$ such that
$X^{\prime} \text{ and } Y^{\prime}$ are distinct,
$\gamma_1^{}  \text{ and } \gamma_2^{}$ belong to
$ MP(\alpha) \text{ and } MP(\beta)$ respectively,
and one of the following symmetric cases hold:
\begin{enumerate}
  
  \item[(a)] $X^{\prime} V, \, Y^{\prime} V \, \in \, RF(\alpha, \gamma_1^{}) $ \emph{and}

  \item[(b)] $X^{\prime} , \, Y^{\prime} \, \in \, RF(\beta, \gamma_2^{})$
  
\end{enumerate} or \begin{enumerate}

  \item[(c)] $X^{\prime} , \, Y^{\prime} \, \in \, RF(\alpha, \gamma_1^{})$ \emph{and}

  \item[(d)] $X^{\prime} V, \, Y^{\prime} V \, \in \, RF(\beta, \gamma_2^{}) $

\end{enumerate}
First of all there are only polynomially many strings
in~$MP(\alpha)$ and $MP(\beta)$. The two cases can be
checked in polynomial time by Lemma~\ref{FindingStringsThruIntersection}. \end{proof}

\begin{figure}
\centering
\epsfig{file=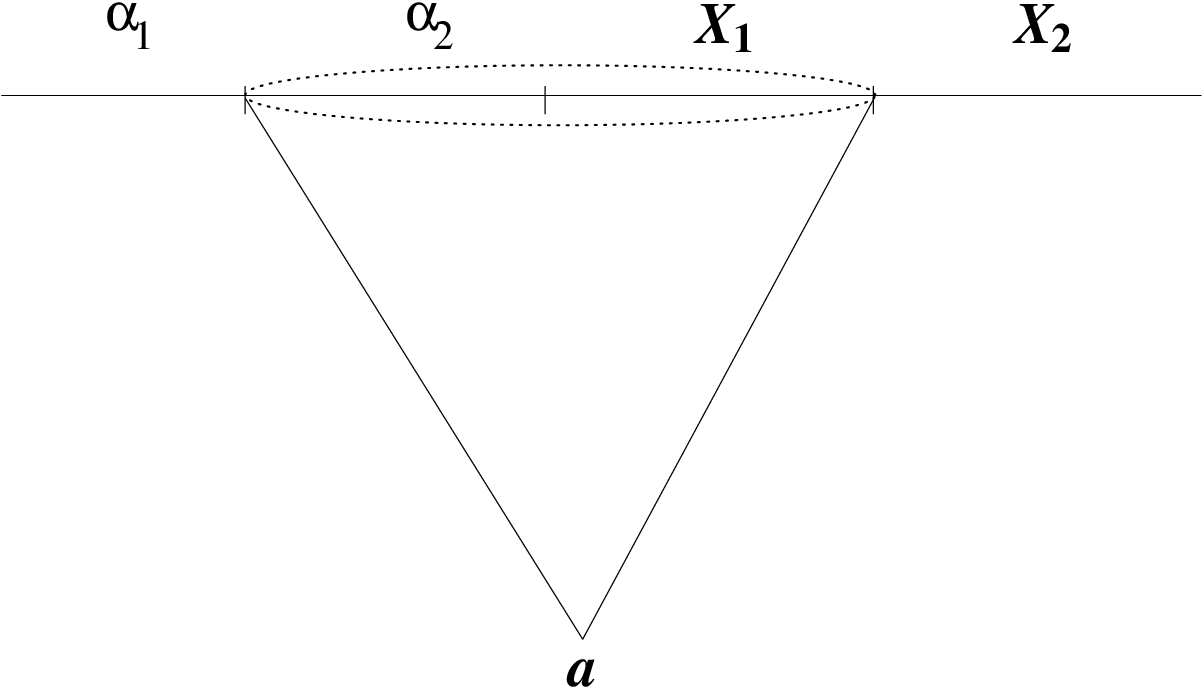, width=4in}
\caption{What the monadic cone looks like for $\alpha$}
\label{monadicCone}
\end{figure}

Figure~\ref{monadicCone} illustrates how $\alpha X$ reduces to its normal form $\alpha^{}_1 a X_2$. ($\alpha$
and $X$ are irreducible strings.)

\pagebreak

\vspace{1cm}
\noindent
Below are the lemmas for the one-mapping case:\\



\begin{lem}{\label{CEOneMappingMonadicLemma5}}
Let $\mu, X, Y \in IRR(R)$ where $X \neq Y$. Then $\mu X  \downarrow  \mu Y$ if and only
if there exist strings~$X^{\prime}, Y^{\prime}, W, \gamma$ such that
\begin{enumerate}

\item $\gamma \in MP(\mu)$,

\item $X = X^{\prime} W$, $Y = Y^{\prime} W$, $X^{\prime} \neq Y^{\prime} \; $ \emph{and}

\item $\mu X^{\prime} \stackrel{!}{{\longrightarrow}_{R}^{}} ~ \gamma \; \; ~ {}_R^{}{\!{\stackrel{!}{\longleftarrow}}} \; \mu Y^{\prime}$.

\end{enumerate}


\end{lem}

\begin{proof}
Let $Z$ be the normal form of $\mu X$ and $\mu Y$. 
Then there exists strings
$\mu_1, \mu_2, \mu_3, \mu_4, X_1,$ $X_2, Y_1, Y_2$ such that
$\mu = \mu_1\mu_2 = \mu_3\mu_4$, $X=X_1X_2$, $Y=Y_1Y_2$ and
\begin{tabbing}
$\qquad$ \= $\mu_{2} X_1 \; \rightarrow_{}^! \; a$,\\
         \> $\mu_{4} Y_1 \; \rightarrow_{}^! \; b, \qquad$ and \\[+5pt]
         \> $Z = {\mu_{1}}\, a \, {X_2} = {\mu_{3}} \, b \, {Y_2}$.\\[-12pt]
\end{tabbing}
where $a,b \in \Sigma \cup \{\lambda\}$.
If $X_1 = Y_1$, then the same reduction can
be applied on both sides, i.e., $\mu_{2} = \mu_{4}$ and~$a = b$.
But the rest of the string $X_2 \neq Y_2$
since $X \neq Y$. Therefore, we conclude that $X_1 \neq Y_1$.  It can
also be seen that $\mu_1 a, \, \mu_3 b \in MP(\alpha)$.

We now consider two cases:
\begin{enumerate}
  \item[(a)] $X_2$ is a suffix of $Y_2$: Let $Y_2 = Y_2^{\prime} X_2$.
  Then ${\mu_{1}}\, a \; = \; {\mu_{3}} \, b \, {Y_2^{\prime}}$. We can take
  $\gamma = {\mu_{1}}\, a$, $X^{\prime} = X_1$, $Y^{\prime} = Y_1 Y_2^{\prime}$
  and $W = X_2$.

  \item[(b)] $Y_2$ is a suffix of $X_2$: Let $X_2 = X_2^{\prime} Y_2$.
  Then ${\mu_{1}}\, a X_2^{\prime} \; = \; {\mu_{3}} \, b$.
  In this case we can take
  $\gamma = {\mu_{3}}\, b$, $X^{\prime} = X_1 X_2^{\prime}$, $Y^{\prime} = Y_1$
  and $W = X_2$. \qedhere
\end{enumerate}
\end{proof}

\vspace{0.1in}


\begin{lem}{\label{CEOneMappingMonadicLemma314}}
\openup 0.5em
  Let $\alpha, \beta \in IRR(R)$. Then there exist \emph{distinct}
  irreducible strings~$X$~and~$Y$ such that $\alpha X \downarrow
  \alpha Y$ and $\beta X \downarrow \beta Y$ if and only if there
  exist irreducible strings~$X^{\prime}, Y^{\prime}, V, W, \gamma_1^{}, \gamma_2^{}$ such that

\openup -0.5em
\begin{enumerate}

\item $X^{\prime} \, \neq \, Y^{\prime}$,

\item $\gamma_1^{} \in MP(\alpha)$,

\item $\gamma_2^{} \in MP(\beta)$,

\item $X = X^{\prime} V W$, $Y = Y^{\prime} VW$, and

\item either \begin{itemize}
  
  \item[(a)] $\alpha X^{\prime} V \stackrel{!}{{\longrightarrow}_{R}^{}} ~ \gamma_1^{} \; \; ~ {}_R^{}{\!{\stackrel{!}{\longleftarrow}}} \;
  \alpha Y^{\prime} V \; \; $ \emph{and}

  \item[(b)] $\; \; \; \beta X^{\prime} \stackrel{!}{{\longrightarrow}_{R}^{}} ~ \gamma_2^{} \; \; ~ {}_R^{}{\!{\stackrel{!}{\longleftarrow}}} \;
  \beta Y^{\prime}$.
  
\end{itemize} or \begin{itemize}

  \item[(c)] $\; \; \; \alpha X^{\prime} \stackrel{!}{{\longrightarrow}_{R}^{}} ~ \gamma_1^{} \; \; ~ {}_R^{}{\!{\stackrel{!}{\longleftarrow}}} \;
  \alpha Y^{\prime} \; \; $ \emph{and}

  \item[(d)] $\beta X^{\prime} V \stackrel{!}{{\longrightarrow}_{R}^{}} ~ \gamma_2^{} \; \; ~ {}_R^{}{\!{\stackrel{!}{\longleftarrow}}} \;
  \beta Y^{\prime} V$.
\end{itemize}

\end{enumerate}

\end{lem}
\begin{proof}
Assume that there exist strings~$X^{\prime}, Y^{\prime}, V,
\gamma_1^{}, \gamma_2^{}$ that satisfy the properties above. Let us
consider the fifth property. It shows that $\alpha X^{\prime} V
\downarrow \alpha Y^{\prime} V$ as well as $\beta X^{\prime} V
\downarrow \beta Y^{\prime} V$.  Now suppose $X = X^{\prime} V$ and $Y
= Y^{\prime} V$. Therefore, we can see that $\alpha X \downarrow
\alpha Y$ as well as $\beta X \downarrow \beta Y$ and $X$ and $Y$ are
distinct irreducible strings.

Conversely, assume that there exist distinct irreducible strings $X$
and $Y$ such that $\alpha X \downarrow \alpha Y$ and $\beta X
\downarrow \beta Y$. We start by considering the case $\alpha X
\downarrow \alpha Y$ such that $X \neq Y$.
By Lemma~\ref{CEOneMappingMonadicLemma5}, there must be strings
$X^{\prime}$, $Y^{\prime}$, $Z$ and~$\gamma_1$ such that $X=X^{\prime} Z$,
$Y=Y^{\prime} Z$ and 
$\alpha X^{\prime} \; \stackrel{!}{{\longrightarrow}_{R}^{}} ~ \gamma_1 ~ \; \; {}_R^{}{\!{\stackrel{!}{\longleftarrow}}} \; \alpha Y^{\prime}$
where~$\gamma_1 \in MP(\alpha)$.
Similarly, for the case $\beta X \downarrow \beta
Y$, $X$ can be written as $X^{\prime\prime} Z^{\prime}$ and $Y$
can be written as $Y^{\prime\prime} Z^{\prime}$ such that 
$\beta X^{\prime\prime} \; \stackrel{!}{{\longrightarrow}_{R}^{}} ~ \gamma_2 \; \; ~ {}_R^{}{\!{\stackrel{!}{\longleftarrow}}} \; \beta Y^{\prime\prime}$ for some $\gamma_2 \in MP(\beta)$.


We have to consider two cases depending on whether $X^{\prime}$ is a
prefix of $X^{\prime\prime}$ or vice versa. It is not hard to see that
they correspond to the two cases in condition~5.
\end{proof}


\begin{cor}{\label{CEOneMappingMonadicLemma3}}
\openup 0.5em
  Let $\alpha, \beta \in IRR(R)$. Then there exist \emph{distinct}
  irreducible strings~$X$~and~$Y$ such that $\alpha X \downarrow
  \alpha Y$ and $\beta X \downarrow \beta Y$ if and only if there
  exist irreducible strings~$X^{\prime}, Y^{\prime}, V, \gamma_1^{}, \gamma_2^{}$ such that

\openup -0.5em
\begin{enumerate}

\item $X^{\prime} \, \neq \, Y^{\prime}$,

\item $\gamma_1^{} \in MP(\alpha)$,

\item $\gamma_2^{} \in MP(\beta)$,

\item $X^{\prime} V$ and $Y^{\prime} V$ are irreducible, and

\item either \begin{itemize}
  
  \item[(a)] $\alpha X^{\prime} V \stackrel{!}{{\longrightarrow}_{R}^{}} ~ \gamma_1^{} \; \; ~ {}_R^{}{\!{\stackrel{!}{\longleftarrow}}} \;
  \alpha Y^{\prime} V \; \; $ \emph{and}

  \item[(b)] $\; \; \; \beta X^{\prime} \stackrel{!}{{\longrightarrow}_{R}^{}} ~ \gamma_2^{} \; \; ~ {}_R^{}{\!{\stackrel{!}{\longleftarrow}}} \;
  \beta Y^{\prime}$.
  
\end{itemize} or \begin{itemize}

  \item[(c)] $\; \; \; \alpha X^{\prime} \stackrel{!}{{\longrightarrow}_{R}^{}} ~ \gamma_1^{} \; \; ~ {}_R^{}{\!{\stackrel{!}{\longleftarrow}}} \;
  \alpha Y^{\prime} \; \; $ \emph{and}

  \item[(d)] $\beta X^{\prime} V \stackrel{!}{{\longrightarrow}_{R}^{}} ~ \gamma_2^{} \; \; ~ {}_R^{}{\!{\stackrel{!}{\longleftarrow}}} \;
  \beta Y^{\prime} V$.
\end{itemize}

\end{enumerate}

\end{cor}
\begin{proof}
The proof follows from the proof of Lemma~\ref{CEOneMappingMonadicLemma314}.
\end{proof}

\begin{thm}{\label{CEOneMappingMonadicLemma4}}
\openup 0.5em
  Let $\alpha, \beta \in IRR(R)$. Then there exist \emph{distinct}
  irreducible strings~$X$~and~$Y$ such that $\alpha X \downarrow
  \alpha Y$ and $\beta X \downarrow \beta Y$ if and only if there
  exist strings~$X^{\prime}, Y^{\prime}, V, \gamma_1^{}, \gamma_2^{}$ such that

\openup -0.5em
\begin{enumerate}

\item $X^{\prime} \, \neq \, Y^{\prime}$,

\item $\gamma_1^{} \in MP(\alpha)$,

\item $\gamma_2^{} \in MP(\beta)$,

\item either \begin{itemize}
  
  \item[(a)] $X^{\prime} V, \, Y^{\prime} V \, \in \, RF(\alpha, \gamma_1^{}) $ \emph{and}

  \item[(b)] $X^{\prime} , \, Y^{\prime} \, \in \, RF(\beta, \gamma_2^{})$
  
\end{itemize} or \begin{itemize}

  \item[(c)] $X^{\prime} , \, Y^{\prime} \, \in \, RF(\alpha, \gamma_1^{})$ \emph{and}

  \item[(d)] $X^{\prime} V, \, Y^{\prime} V \, \in \, RF(\beta, \gamma_2^{}) $

\end{itemize}

\end{enumerate}

\end{thm}
\begin{proof}
    The proof follows from the above lemmas, corollaries and by the definition of right factor.
\end{proof}

\begin{lem}
    \label{unrestricted_rf_regular}
    Let $T$ be monadic, finite, and convergent string rewriting system. Let $x,y \in IRR(T)$. Then $RF(x,y)$ is a regular language and a DFA for it can be constructed in polynomial time.
\end{lem}

\begin{proof}
    A DFA accepting $RF(x,y)$ will be given by $M = (Q, \Sigma, \delta, q_0, F)$ where the set of states $Q \subset (MP(x) \times IRR(T)) \cup \{ \bot \}$ and $\bot$ will be used to label the unique dead state of~$M$. Let $q_0 = (x, \lambda) \in Q$ be the start state of $M$. The transition function is then
    $$
    \delta((\alpha, \beta), a) := 
    \begin{cases}
        \bot &{\beta}a \not \in IRR(T) \\
        ({\alpha}{\beta}a\!\downarrow, \, \lambda) &\beta a,\; \alpha \beta a \not \in IRR(T) \\
        (\alpha, {\beta}a) &\alpha \beta a \in IRR(T),\; \beta a \text{ proper substring of a lhs} \\
        (\alpha, \beta) &\alpha \beta a \in IRR(T),\; \beta a \text{ not a proper substring of a lhs},\; \alpha \beta a \text{ prefix of } y \\
        \bot &\text{otherwise}
    \end{cases}
    $$
    We then set $F = \{(u,v) \in MP(x) \times IRR(T)  \; | \; y = uv\}$. Correctness of the construction follows from the fact that for all $(z_1, z_2) \in Q$ we have that $z_1z_2 \in IRR(T)$ and that $z_2$ is a substring of either the left-hand sides of rules in $T$ or of $y$. 
    
    Additionally, $M$ can be constructed in polynomial time as $x,y$ are fixed and a DFA for IRR(T) can be constructed in polynomial time.  
\end{proof}

\section*{Conclusion and Future Work}

Inspired by past works done on string rewriting systems, we explored the Fixed Point, Common Term, and Common Equation problems. For these problems we looked at string rewriting systems that were convergent, length-reducing, dwindling, and monadic. We provided complexity results for the Fixed Point problem for convergent and dwindling SRSs, for the Common Term problem we provided complexity results for dwindling SRSs, and for the Common Equation problem we provided complexity results for dwindling and monadic SRSs. 

For the sake of brevity and clarity, we only discussed string rewriting systems in this paper. Our future work will include the investigation of these problems for general term rewriting systems. Additionally, we will also explore forward closed SRSs.

\bibliographystyle{alpha}
\bibliography{unifs-duals-journal}

\end{document}